\documentclass[runningheads]{llncs}
\usepackage{graphicx}

\usepackage{hyperref}

\graphicspath{{Figures/}{pictures/}{images/}{./}} 

\renewcommand{\paragraph}[1]{\smallskip\noindent\textit{#1}}

\usepackage{wrapfig}
\usepackage{color,cite,xspace}
\usepackage{marginnote}
\renewcommand{\marginnote}[1]{}

\newcommand{\demersabbrev}{DC\xspace}
\newcommand{\demersabbrevs}{DCs\xspace}

\usepackage{multirow}
\usepackage[table,xcdraw]{xcolor}
\usepackage{amssymb,amsmath}
\usepackage{mathtools} 

\usepackage{thm-restate}

\usepackage[utf8]{inputenc}

\begin{document}
\title{Computing Stable Demers Cartograms%
\thanks{This research was initiated at NII Shonan Meeting 127 ``Reimagining the Mental Map and Drawing Stability''. 
M.~Sondag is supported by The Netherlands Organisation for Scientific Research (NWO) under project no.~639.023.20. S.~Kobourov is supported by NSF grants CCF-1740858, CCF-1712119
and DMS-1839274.}%
}
%
%

\author{Soeren Nickel\inst{1}\and
Max Sondag\inst{2} \and
Wouter Meulemans\inst{2} \and
Markus Chimani\inst{3} \and
Stephen Kobourov\inst{4} \and
Jaakko Peltonen\inst{5}\and
Martin Nöllenburg\inst{1}}

\authorrunning{S. Nickel et al.}
%
\institute{TU Wien, Austria  \email{\{soeren.nickel@,noellenburg@ac.\}tuwien.ac.at} \and
TU Eindhoven, Netherlands \email{\{m.f.m.sondag,w.meulemans\}@tue.nl} \and
University of Osnabrück, Germany \email{markus.chimani@uni-onsabrueck.de} \and
University of Arizona, USA \email{kobourov@cs.arizona.edu} \and
Tampere University, Finland \email{jaakko.peltonen@tuni.fi}}
\maketitle              
\begin{abstract}
Cartograms are popular for visualizing numerical data for map regions.
Maintaining correct adjacencies is a primary quality criterion for cartograms. 
When there are multiple data values per region (over time or different datasets) shown as animated or juxtaposed cartograms, preserving the viewer's mental-map in terms of stability between cartograms is another important criterion. 
We present a method to compute stable Demers cartograms, where each region is shown as a square and similar data yield similar cartograms. 
We enforce orthogonal separation constraints with linear programming, and measure quality in terms of keeping adjacent regions close (cartogram quality) and using similar positions for a region between the different data values (stability). 
Our method guarantees ability to connect most lost adjacencies with minimal leaders.
Experiments show our method yields good quality and stability. 

\keywords{Time-varying data \and Cartograms \and Mental-map preservation}
\end{abstract}

\section{Introduction}

Myriad datasets are georeferenced and relate to specific places or regions. A natural way to visualize such data in their spatial context is by cartographic maps.
A choropleth map is a prominent tool, which colors each region in a map by its data value. 
Such maps have several drawbacks: data may not be correlated to region size and hence the visual salience of large vs small regions is not equal.
Moreover, colors are difficult to compare and not the most effective encoding for numeric data~\cite{munzner2014visualization}, requiring a legend to facilitate relative assessment.

Cartograms, also called value-by-area maps, overcome the drawbacks by reducing spatial precision in favor of clearer encoding of data values: 
the map is deformed such that each region's visual size is proportional to its data value. 
Attention is then drawn to items with large data values and comparison of relative magnitudes becomes a task of estimating sizes -- which relies on more accurate visual variables for numeric data~\cite{munzner2014visualization}.
This also frees up color as a visual variable. 
Cartogram quality is assessed by 
criteria~\cite{nusrat2016state} including
\textbf{1. Spatial deformation:} regions should be placed close to their geographic position;
\textbf{2. Shape deformation:} each region should resemble its geographic shape;
\textbf{3. Preservation of relative directions:} spatial relations such as north-south and east-west should be maintained.
\textbf{4. Topological accuracy:} geographically adjacent regions should be adjacent in the cartogram, and vice versa.
\textbf{5. Cartographic error:} relative region sizes should be close to the data values. 
Criteria 1-4 describe geographical accuracy of the region arrangement. 
Maintaining relative directions also helps preserve a viewer's spatial mental model~\cite{tversky1993cognitive}
Criterion 5 (also called statistical error) captures how well data values are represented. Often techniques aim at zero cartographic error sacrificing other criteria.

Cartograms can also be effective for showing different datasets of the same regions, arising from time-varying data such as yearly censuses yielding temporally ordered values for each region, or from available measurements of different demographic variables that we want to explore, compare and relate, yielding a vector or set of values for each region. 
Visualizations for multiple cartograms include animations (especially for time series), small multiples showing a matrix of cartograms, or letting a user interactively switch the mapped value in one cartogram.
See for example the interactive Demers cartogram accompanying an article from the New York Times\footnote{\url{https://archive.nytimes.com/www.nytimes.com/interactive/2008/09/04/business/20080907-metrics-graphic.html}, accessed June 2019.}.
In such methods, cartograms should be as similar as the data values allow: we thus want cartograms to be \emph{stable} by using similar layouts.
This helps retain the viewer's mental map\cite{misue1995layout}, supporting linking and tracking across cartograms. 
Thus, we obtain an important criterion with multivariate or time-varying data.
\textbf{Stability}: for high stability, cartograms for the same regions using different data values should have similar layouts.
The relative importance of the criteria depends on the tasks to be facilitated. 
Nusrat and Kobourov's taxonomy of ten tasks~\cite{nusrat2016state} can also be considered with multiple cartograms.
Many tasks focus on the data values. 
As such, a representation of a region of low complexity allows for easier estimation and size comparison.


\paragraph{Contribution.}
We focus on \emph{Demers cartograms} (\demersabbrev; \cite{Demers}) which represent each region by a suitably
sized square, similar to Dorling cartograms~\cite{dorling96} which use circles.
Their simplicity allows easy comparison of data values, since aspect ratio is no longer a factor, unlike, e.g., for rectangular cartograms~\cite{ks07}.
However, as abstract squares incur shape deformation, in spatial recognition tasks the cartogram embedding as a whole must be informative, so the layout must optimize as much as possible the other geographic criteria: \emph{spatial deformation}, \emph{preservation of relative directions} and \emph{topological accuracy}.
We contribute an efficient linear programming algorithm to compute high-quality stable \demersabbrevs.
Our \demersabbrevs have no cartographic error, satisfy given constraints on spatial relations, and allow trade-off between topological error and stability. 
Linear interpolation between different \demersabbrevs yields no overlap during transformation. 
Lost adjacencies--satisfying a mild assumption--can be shown as minimal-length planar orthogonal lines.
Fig.~\ref{fig:teaser} shows examples.
Experiments compare settings of our linear program to each other and to a force-directed layout we introduce (also novel for \demersabbrevs); results show that our linear program efficiently computes stable \demersabbrevs.

\begin{figure}[t]
  \centering

  \includegraphics[width=0.24\linewidth]{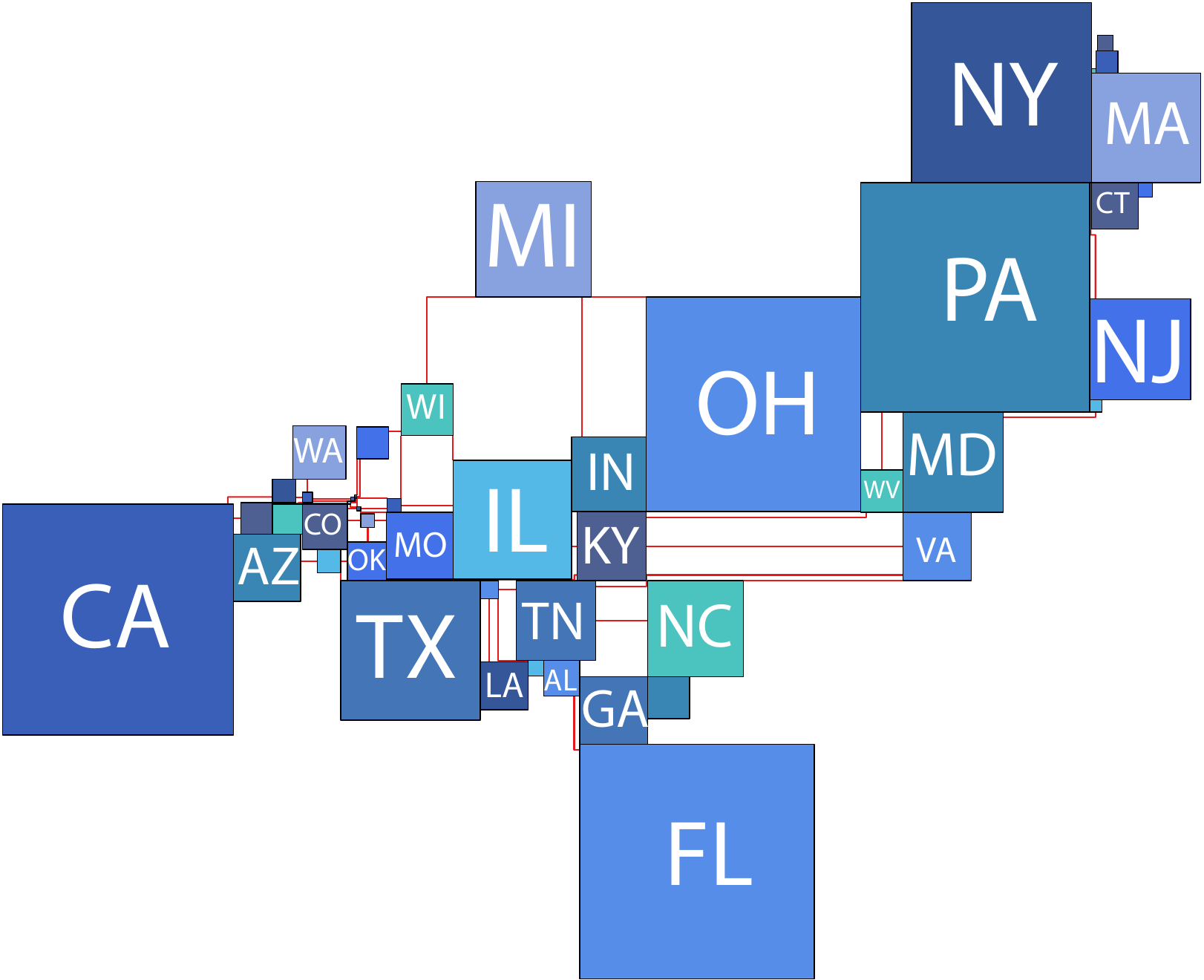}
  \includegraphics[width=0.24\linewidth]{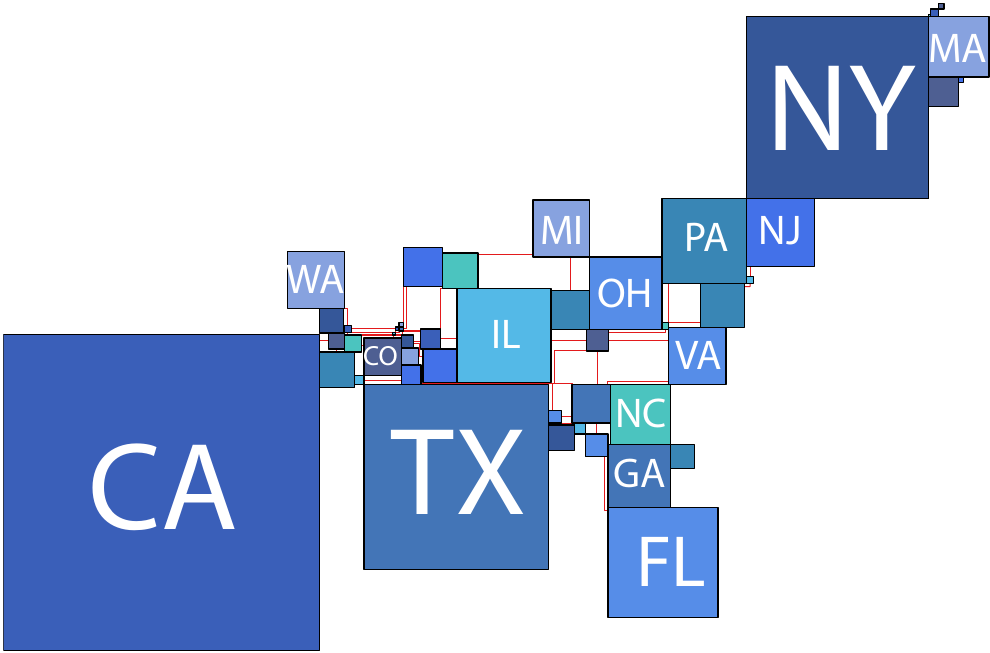}
  \includegraphics[width=0.24\linewidth]{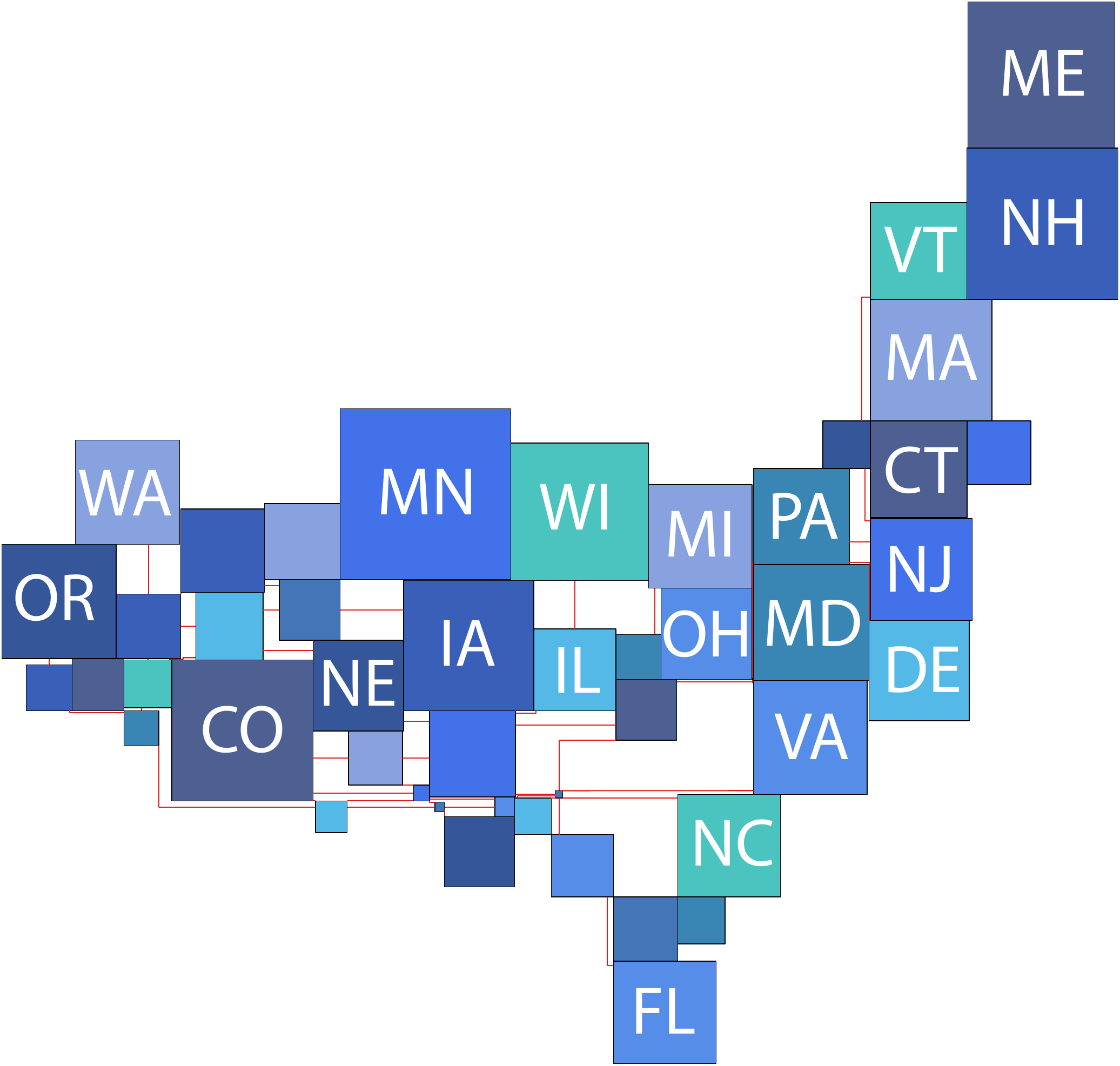}
  \includegraphics[width=0.24\linewidth]{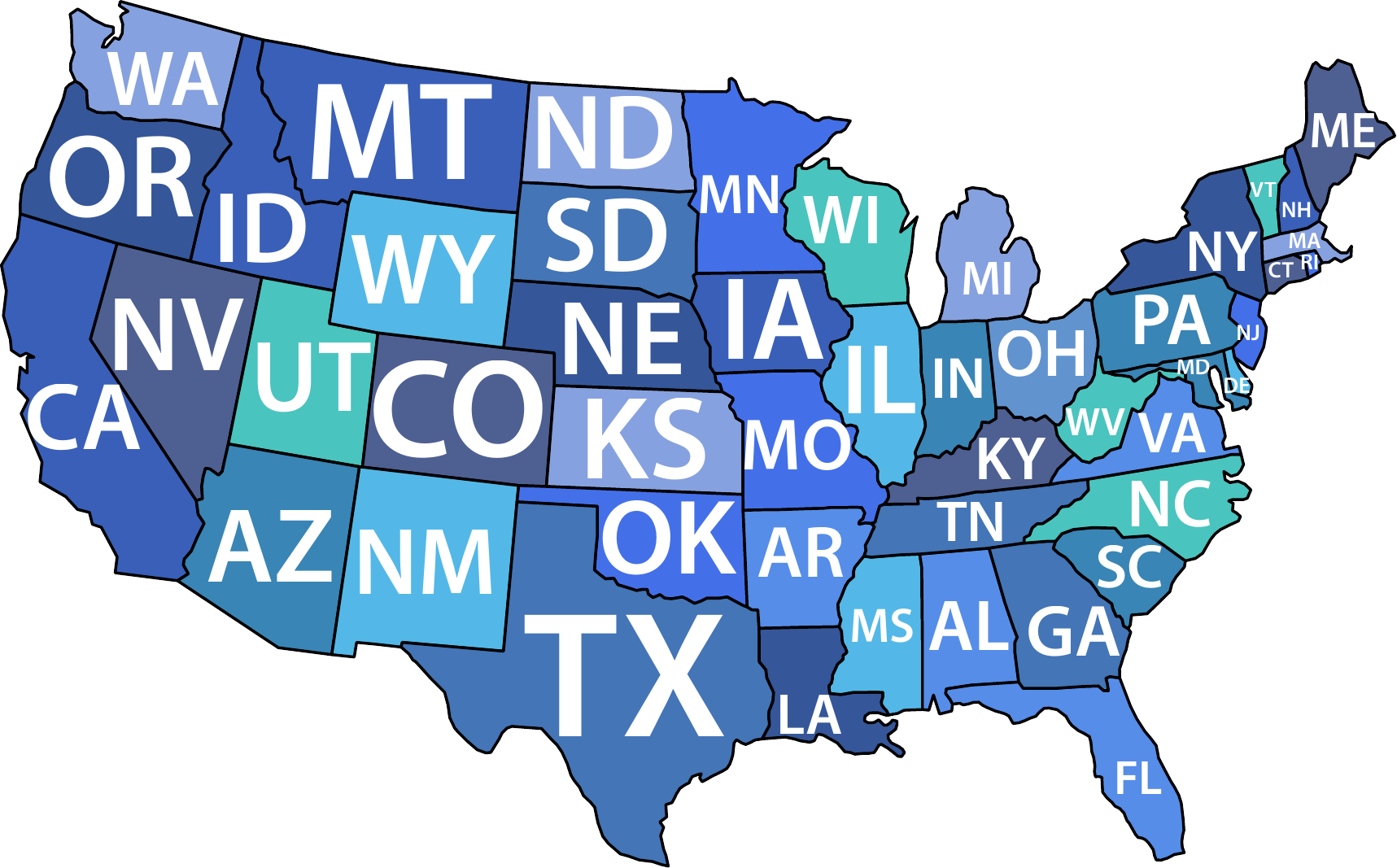}
  \caption{Cartograms displaying drug poisoning mortality, total GDP and population of the contiguous states of the US in 2016. The layout minimizes distance between adjacent regions. Lost adjacencies are indicated with red leaders. Color is used only to facilitate correspondences between the cartograms.}
	\label{fig:teaser}
\end{figure}	

\paragraph{Related work.}
Cartogram-like representations date to the 1800s. In the 1900s most standard cartogram types were defined, including rectangular value-by-area cartograms 
\cite{Raisz34} 
and more recent ones 
\cite{hkps04,ks07}.
The first automatically generated cartograms are continuous deformation ones 
\cite{Tobler73} 
followed by others 
\cite{HK98,GN04}.
Dorling cartograms~\cite{dorling96} and \demersabbrevs~\cite{Demers} exemplify the non-contiguous type representing regions by circles and squares respectively. 
Layouts representing regions by rectangles and rectilinear polygons have received much attention in algorithmic literature, see e.g. \cite{alam2013computing,buchin2012evolution,eppstein2009area}, 
and typically focus on aspect ratio, topological error and region complexity.
Compared to \demersabbrevs, rectilinear variants have higher visual complexity and added difficulty to assessing areas.
No cartogram type can guarantee a both statistically and geographically accurate representation; see a recent survey~\cite{nusrat2016state}. 
Measures exist to evaluate quality of cartogram types and algorithms, see e.g. \cite{alam2015quantitative,keim2004cartodraw}.

There is little work on evaluating or computing stable cartograms for time-varying or multivariate data. Yet they are used in such manner, e.g., as a sequence of contiguous cartograms showing the evolution of the Internet \cite{johnson2015analyzing}. 

\demersabbrevs relate to contact representations, encoding adjacencies between neighboring regions as touching squares. The focus in graph theory and graph drawing literature lies on recognizing which graphs can be perfectly represented. Even the unit-disk case is NP-hard 
\cite{breu1998unit}, 
though efficient algorithms exist for some restricted graph classes \cite{di2015low}. Klemz et al.~\cite{klemz2015recognizing} consider a vertex-weighted variant using disks, that is, with varying disk sizes. 
Various other techniques are similar to \demersabbrevs, using squares or rectangles for geospatial information. Examples include grid maps, see, e.g., \cite{eppstein} for algorithms and \cite{smwg} for computational experiments.
%
Recently, Meulemans \cite{diamondoverlap} introduced a constraint program to compute optimal solutions under orthogonal order constraints for diamond-shaped symbols.
We use similar techniques, but refer to Section~\ref{sec:single} for a discussion of the differences.



\section{Computing a single \demersabbrev}
\label{sec:single}

First, we consider a \demersabbrev\ for a single weight vector. 
We are given a set of weighted regions with their adjacencies, and a set of directional relations. 
We compute a layout realizing the weights with disjoint squares that may touch only if adjacent, so that directional relations are ``roughly'' maintained. 
We quantify the quality of the layout by considering the distances between any two squares representing adjacent regions. 
We show that the problem, under appropriate distance measures, can be solved via linear programming in polynomial time.

\paragraph{Formal setting.}
We are given an input graph $G=(R,T)$. For each \emph{region} $r\in R$ we are given its centroid in $\mathbb{R}^2$ and its weight $w(r)$, the side length of the square that represents it in output.
The graph has an edge in $T$ if and only if the original regions are adjacent, thus their respective squares in the output \emph{should} be adjacent as well.
We are also given two sets $H$, $V$ of ordered region pairs. 
A pair $(r,r')$ is in $H$, if $r$ should be horizontally separated from $r'$ such that there exists a vertical line $\ell$ with the square of $r$ being left of $\ell$ and $r'$ to its right. 
Analogously, $V$ encodes vertical separation requirements.
If $r$ and $r'$ are adjacent, then $(r,r')$ is either in $H$ or in $V$ (but not in $H \cap V$) and they \emph{should} touch $\ell$, otherwise we require a strict separation to avoid false adjacencies; 
we are given a minimum gap $\varepsilon$ to ensure that this non-adjacency can be visually recognized.\footnote{In the implementation, $\varepsilon$ is the minimum of the side length of the smallest region and 5\% of the diagonal of the bounding box of the input regions~$R$.} 
The sets $H$ and $V$ model the relative directions criterion for \demersabbrevs and any two regions are paired in at least one of those sets. 
To ensure a \demersabbrev exists satisfying the separation constraints, the directed graph $D = (R, H \cup V)$ must be a directed \emph{acyclic} graph (DAG).
We consider these relations transitive: if $(r,r') \in H$ and $(r',r'') \in H$, then this enforces that there exists a vertical line separating $(r,r'')$ in any \demersabbrev and thus $(r,r'')$ is in $H$.

The output---a placement of a square for each region---can be stored as a point $P \colon {R} \rightarrow \mathbb{R}^2$ for each region, encoding the center of its square.
A placement $P$ is \emph{valid}, if it satisfies the separation constraints of $H$ and $V$. 
This implies all squares are pairwise interior disjoint (or fully disjoint for nonadjacent regions).
We look for a valid placement where distances between non-touching squares of originally adjacent regions are minimized; this will be made more precise below.

\paragraph{Deriving separation constraints.} 
The regions' weights are given and their adjacencies and centroids easily derived, but  separation constraints $H$ and $V$ are not. 
Various models can determine good directions or separation constraints~\cite{bkssv-slmdr-11}. 
We use the following model; it is symmetric and ensures constraints 
form a DAG.

For two regions $(r,r')$ represented by centroids, we check whether their horizontal or vertical distance is larger. 
In the former case, we add $(r,r')$ to $H$ if $r$ is left of $r'$ and $(r',r)$ to $H$ otherwise. 
In the latter case, we add the pair to $V$ in the appropriate order. 
We call this the \emph{weak setting}.
We call constraints added in this setting \emph{primary separation constraints}.

In the \emph{strong setting}, we may add an extra constraint for nonadjacent region pairs whose bounding boxes admit both horizontal and vertical separating lines: 
if a pair has a primary separation constraint in $H$ or $V$, we add a \emph{secondary} separation constraint to $V$ or $H$ respectively.

\paragraph{Linear Program.}
We model optimal solutions to the problem via a polynomially-sized linear program (LP), which lets us solve the problem in polynomial time.
For each $r\in R$, we introduce variables $x_r$ and $y_r$ for the center $P(r) = (x_r, y_r)$ of the square.
For any originally adjacent regions $\{r,r'\}\in T$ we introduce variables $h_{r,r'}$ and $v_{r,r'}$ for the (non-negative) distance between two squares.
For any two regions $r,r'$, we define shorthands: let $w_{r,r'}:=\frac{(w(r)+w(r'))}{2}$ and let $\mathit{gap}_{r,r'}=\varepsilon$ if $\{r,r'\}\not\in T$, and $0$ otherwise.
\begin{align}
    &\min \sum_{\{r,r'\} \in {T}} h_{r,r'} + v_{r,r'} && \label{eq:objective}\\
    &x_{r'} - x_{r} \geq w_{r,r'} + \mathit{gap}_{r,r'} && \forall (r,r')\in H \label{eq:sepX}\\
    &y_{r'} - y_{r} \geq w_{r,r'} + \mathit{gap}_{r,r'} && \forall (r,r')\in V \label{eq:sepY}\\
    &h_{r,r'} \geq \max\{ (x_r - x_{r'}) - w_{r,r'}, (x_{r'} - x_r) - w_{r,r'} \} && \forall \{r,r'\}\in T \label{eq:measureX}\\
     &v_{r,r'} \geq \max\{ (y_r - y_{r'}) - w_{r,r'}, (y_{r'} - y_r) - w_{r,r'} \} && \forall \{r,r'\}\in T \label{eq:measureY}\\
    &h_{r,r'}, v_{r,r'} \geq 0 && \forall \{r,r'\}\in T \label{eq:measureNonZero}
\end{align}
%
%
The objective \eqref{eq:objective} minimizes a sum of the distances between regions with broken adjacencies in the $L_1$ metric.
Constraints \eqref{eq:sepX} and \eqref{eq:sepY} ensure separation requirements by forcing square centers far enough apart.
For nonadjacent regions, the $\mathit{gap}$ function assures a recognizable gap of width $\varepsilon$ between resulting squares.
Constraints \eqref{eq:measureX}--\eqref{eq:measureNonZero} bind distance variables $h,v$ with positional variables $x,y$. 
Here \eqref{eq:measureX} and \eqref{eq:measureY} encode two linear constraints per line, one for each term in the  `$\max$' function.
As \eqref{eq:objective} minimizes the distances, it suffices to enforce lower bounds,  hence the `$\ge$' in the constraints. 
In an optimal solution, either one of the two versions, or the non-negativity constraint \eqref{eq:measureNonZero} will be satisfied with equality.

\paragraph{Improving the gaps.}
The above model has two minor flaws. 
\emph{First,} two squares `touch' even if they only do so at corners; we resolve this by adding $\varepsilon$ to the right-hand side of \eqref{eq:measureX} (or \eqref{eq:measureY}) for vertically (or horizontally, respectively) separated region pairs in~$T$.
This allows $h_{r,r'} = 0$ ($v_{r,r'}=0$), when squares share a segment at least $\varepsilon$ long.
\emph{Second,} in the strong setting the LP asks for a minimum gap $\varepsilon$ along both axes. This is not not needed for visual separation, so we remove the gap requirement from the secondary separation constraint.

\paragraph{Fine-tuning the optimization criteria.}
The LP minimizes a sum of distances between adjacent regions. 
Cartogram literature emphasizes counting lost adjacencies between regions, not the distance between them.
We prefer our measure since 
\emph{1)} there is a big difference if two neighboring countries are set apart by a small or large gap; \emph{2)} while the LP can be turned to an \emph{integer} linear program to count lost adjacencies, it greatly increases computational complexity---optimizing for adjacencies is typically NP-hard, e.g., for disks~\cite{breu1998unit,bdlrst-rscplrudct-15} 
or segments~\cite{h-cglsn-01}.

Our linear program typically admits several optimal solutions, due to translation invariance and since touching squares may slide freely along each other as long as they touch.
We introduce a \emph{secondary} term to the objective 
to nuance selection of better layouts, multiplied by a \emph{small} constant to not interfere with the original (primary) objective.
The secondary term optimizes preservation of relative directions between squares within the freedom of the optimal solution.

Consider regions $r$ and $r'$. 
W.l.o.g., assume their original centroids are horizontally farther apart than vertically, and $r$ is left of $r'$, so $(r,r')\in H$.
We compute a directional deviation $d_{rr'} = |(y_r + \alpha (x_{r'} - x_r)) - y_{r'}|$, where $\alpha$ is the (finite) slope of the ray from $r$ to $r'$ in the input graph $G$.
Similar to \eqref{eq:measureX}, the objective function will minimize $d_{rr'}$; we weigh this term more heavily for adjacent regions. 
We thus turn the above formula into two linear inequalities.

Alternatives exist for the secondary criterion: displacement from the original location helps find layouts maintaining many adjacencies for grid maps of equal-size squares \cite{eppstein,smwg}.
For each region we measure $L_1$ displacement from its \emph{origin} (centroid of the original region in the geographic map) to the square center $P(r)$. 

\paragraph{Comparison to overlap removal.}
A technique placing disjoint squares exists to remove overlap of diamond (45 degree rotated square) glyphs for spatial point data~\cite{diamondoverlap}, asking to minimally displace varying-size diamonds to remove all overlap, constrained to keep orthogonal order of their centers. 
Rotating the scenario to yield squares does not yield axis-parallel order constraints but ``diagonal'' ones, different from our strong setting.
A ``weak order constraints'' variant is mentioned, related to our LP in the weak setting, if we change our objective to one only optimizing displacement relative to original locations. 
Fig.~\ref{fig:feasibleRegions} shows similarities and differences considering the feasibility area between two regions.
Extensions in \cite{diamondoverlap} can be applied in our scenario, e.g., reducing actively considered separation constraints by removing transitive relations (``dominance'' in \cite{diamondoverlap}). 
Time-varying data is briefly considered in \cite{diamondoverlap}, only conceptualizing a trade-off between origin-displacement and stability for artificial data;
we discuss several optimization criteria, also focusing on adjacencies which are not considered in \cite{diamondoverlap},  use real-world data experiments, and compare to a baseline \demersabbrev implementation to move beyond the limits of linear programming.

\begin{figure}[t]
\begin{minipage}[t]{0.54\textwidth}
\includegraphics[width=\textwidth]{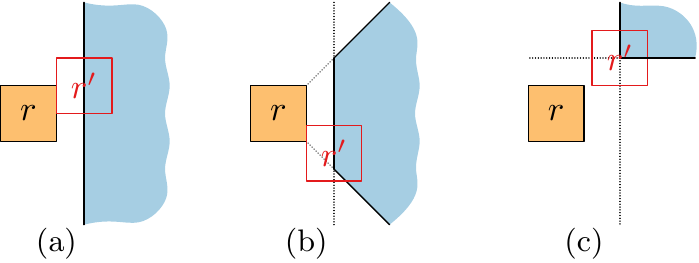}
\end{minipage}
\hfill
\begin{minipage}[t]{0.45\textwidth}
\vspace{-7.5\baselineskip}
\caption{Feasibility area where $r'$ may be placed w.r.t $r$, when $r'$ is primarily to the right of $r$. (a)~In terms of feasibility, our weak setting and weak order constraints in \cite{diamondoverlap} coincide. (b)~Feasibility using a rotated orthogonal order \cite{diamondoverlap}. (c)~Feasibility in the strong setting.}
    \label{fig:feasibleRegions}
\end{minipage}
\end{figure}

The lemma below matches an observation from \cite{diamondoverlap} that carries over to our setting. It implies that cartograms for different weight functions but with the same constraints have a smooth and simple transition between any such \demersabbrevs helping to retain the user's mental map.

\begin{restatable}{lemma}{niceinterpolations}
\label{lem:niceinterpolations}
Let $R$ be a set of regions with separation constraints $H$ and $V$. Let $A$ and $B$ be two \demersabbrevs for $R$, both satisfying $H$ and $V$. Then, any linear interpolation between $A$ to $B$ also satisfies $H$ and $V$ and is thus overlap-free.
\end{restatable}


\section{Computing stable \demersabbrevs for multiple weights}
\label{sec:timevar}

The method can be extended for regions having multiple weights. 
We are given a set of weight functions $W= \{w_1, \dots, w_k\}$. 
We aim to compute a \demersabbrev for each $w_i \in {W}$, i.e., positions $P_i(r)$ for each $r \in R$ and $w_i \in W$.
If each weight function represents the same data semantic, say population size, at different times, we consider ${W} = \{ w_1, \ldots, w_k \}$ ordered by the $k$ time steps; we call this setting \emph{time series}. 
If each weight function represents measurements of different data semantics (possibly at the same time), say population and gross domestic product, we treat ${W}$ as an unordered set; we call  this setting \emph{weight vectors}. 


As we focus on cartogram stability over multiple datasets, we combine the weight functions into one LP that computes the set of \demersabbrevs, with potentially different centers $P_i(r)$ for each region $r$ and weight function $w_i$. 
This lets us add constraints and optimization objectives for stability. 
We change objective~\eqref{eq:objective} and add constraints to minimize displacement between centers of the same region for different weight functions. 
We re-use notation in Section~\ref{sec:single} with superscript $i$ denoting respective variables for weight function $w_i \in W$.
\begin{align}
    &\min \sum_{i=1}^{k}\sum_{\{r,r'\} \in {T}} (h^i_{r,r'} + v^i_{r,r'}) + 
    \mathrlap{ \sum_{\{i,j\} \in I} \sum_{r \in R} (c_r^{i,j} + d_r^{i,j})} && \label{eq:objective2}\\
    &c_r^{i,j} \ge \max \{(x_r^i - x_r^j), (x_r^j - x_r^i)\} && \forall r \in R, \{i,j\} \in I \label{eq:hordisp}\\
    &d_r^{i,j} \ge \max \{(y_r^i - y_r^j), (y_r^j - y_r^i)\} && \forall r \in R, \{i,j\} \in I \label{eq:verdisp}
\end{align}
Here set $I$ contains index pairs of weight functions $\{w_i, w_j\}$ for which  displacement should be minimized. 
For each $r \in R$, variables $c_r^{i,j}$ and $d_r^{i,j}$ measure the horizontal and vertical displacement between $P_i(r)$ and $P_j(r)$ due to \eqref{eq:hordisp} and \eqref{eq:verdisp}. 

For which weight functions to relate in $I$, we consider two options: 1) relate \emph{all pairs} of functions so $I=\binom{W}{2}$, which is natural for weight vectors where an analyst may want to compare the \demersabbrevs for any two weight functions; and 2) relating \emph{consecutive} pairs in a predefined order of the functions so $I=\{(i,i+1) \mid 1 \le i \le k-1\}$, which is natural for time series.
An alternative 3) for time series initially computes a \demersabbrev for $w_1$ (e.g., minimizing displacement to region centroids in the initial map) and then iteratively solves the LP for one \demersabbrev and weight function $w_i$ ($i\ge 2$), where we minimize the displacement only with respect to the previously solved \demersabbrev for weight function $w_{i-1}$. 
Due to its restricted solution space 3) is expected to be faster to solve than 2) but with lower stability. 
In some scenarios another option 4) may be worthwhile: one weight function, say $w_1$, may be considered central to the dataset and displacements are only considered
relative to it, so $I$ contains pairs $\{1,i\}$ for all $2\le i\le k$.

Not all planar graphs can be represented using touching squares of any size.
A real-world example is Luxembourg having three pairwise neighbors; the input graph $G$ is a $K_4$.
Thus any \demersabbrev~may need to break some adjacencies. 
To show lost adjacencies we use \emph{leaders} -- orthogonal polylines connecting the two squares.
We want leaders to have minimal length and low complexity which we can guarantee under mild assumptions: 1) leaders can coincide with square boundaries; 2) regions to be connected are \emph{realisable}, i.e., a valid \demersabbrev (with possibly different weights) exists for each pair of regions such that they are adjacent.
Let $L_1^B(r_1,r_2)$ denote the minimal $L_1$ distance between squares of regions $r_1$ and $r_2$ in \demersabbrev $B$. The following lemmas are proven in Appendix~\ref{app:leaderlength} -- the proof of the first is constructive and gives a simple $O(n^2)$ algorithm to compute all leaders.

\begin{restatable}{lemma}{leaderlength}
\label{lem:leaderlength}
Consider \demersabbrevs with separation constraints $H$, $V$ and two regions $\{r_1,r_2\} \in T$. Let $(r_1,r_2)$ be a minimal pair in $H$ or $V$. 
Then, in any \demersabbrev $B$, there is a monotone leader $\ell$ between $r_1$ and $r_2$ with length $L_1^B(r_1,r_2)$.
\end{restatable}
\begin{restatable}{lemma}{leaderbends}
\label{lem:leaderbends}
Let $\{r_1,r_2\} \in T$ and assume a \demersabbrev~$A$ exists with $r_1$ and $r_2$  adjacent, from which $H$ and $V$ are derived in the strong setting.
Then, for any \demersabbrev~$B$ satisfying $H$ and $V$, a leader $\ell$ exists between $r_1$ and $r_2$ with at most two bends.
\end{restatable}

\section{Experimental setup}
\label{sec:setup}
We compare 18 variants of our linear programs with each other and to 4 variants of a baseline force-directed \demersabbrev~layout implementation, as described below.

\paragraph{Linear programs.}
We categorize our method according to three criteria: A) optimization term, B) method of deriving constraints, and C) how we deal with different time steps.
For A) our linear program admits three primary optimization terms:
\emph{TOP} -- distance between topologically adjacent regions; 
\emph{CNT} -- number of lost adjacencies; 
\emph{ORG} -- distance to the origin (region's centroid in the geographic map).
We use the indicated primary optimization term, complemented by the secondary constraint of maintaining relative directions.
For B), separation constraints are deduced from the input map in one of two ways, \emph{S} and \emph{W}, matching the strong and weak case respectively.
For C), we deal with different weight values (time series/weight vectors) in three ways
called \emph{stability implementations}:
\emph{CO} -- we add an optimization term to minimize distance between layouts of all (complete) weight value pairs;
(2) \emph{SU} -- we add an optimization term to minimize distance between layouts of successive weight values;
(3) \emph{IT} -- we iteratively solve a linear program including an optimization term to minimize distance to previously calculated layouts.
We specify our methods by concatenating the three aspects in order, for example, TOP-S-SU indicates the linear program optimized for distances of topologically adjacent regions with strong separation constraints and with successive weight values linked.

\paragraph{Force-directed method.}
\demersabbrevs are hard to track down in literature, especially regarding computation.
To our knowledge, there is no common baseline for computing a \demersabbrev; we introduce a simple one. As Dorling cartograms and \demersabbrevs are 
similar \cite{Demers} and
Dorling cartograms use a force-directed method, we implement one here, too: FRC. 
For each pair of regions we define a disjointness force based on Chebyshev distance between their centers, which grows quadratically to push squares apart. We use the same desired distance as in Section~\ref{sec:single} at which this force becomes zero.
We also add a force for cartogram quality, either towards their original locations (FRC-O) or between adjacent regions (FRC-T).
We initialize the process with map locations (U; unstable) or the result for previous weights (S; stable). See Appendix~\ref{app:force} for more details.

\paragraph{Metrics: cartogram quality.}
Our algorithms inherently yield zero cartographic error, and shape deformation is constant over all possible \demersabbrevs.
To evaluate cartogram quality we use three metrics, each normalized between 0 and 1
; smaller values are better.
We measure topological accuracy as the number of lost adjacencies (\textbf{MADJ}) in each of the $k$ computed layouts, normalized by the number of adjacencies $k|T|$.
To measure preservation of relative directions (\textbf{MREL}) with respect to the input map, we use the Relative Position Change Metric~\cite{sondag17} which captures the preservation of the spatial mental model (orthogonal order) in a fine-grained way. 
Each rectangle defines eight zones by extending its sides to infinite lines. %
Between a pair of input map regions $(r,r')$ we consider fractions of the bounding box that fall into each zone; if bounding boxes overlap, we scale values so they sum to $1$.
We do the same between the corresponding squares in the cartogram layouts. The measure between two regions is half the sum over all absolute differences between fractions per zone; the value is in $[0,1]$ but is not symmetric. Finally, we take the average over all pairs.
For spatial deformation we measure distance to map origins (\textbf{MDIS}), average $L_1$ distance of each region $r$ in the \demersabbrev to its origin (centroid of $r$ in the geographic map), normalized by dividing with the $L_1$ distance of the diagonal of the map.

\paragraph{Metrics: stability.}
We also want to assess stability, or layout similarity, between the \demersabbrevs by two quality metrics, based on treemap stability metrics \cite{sondag17}, interpreting \demersabbrevs as special treemaps with added whitespace. 
The first is based on geometric distances between the layouts: the layout distance (\textbf{SDIS}) focuses on the change in position of the squares. 
The layout distance change function as presented by Shneiderman and Wattenberg~\cite{shneiderman2001ordered} is the most common one. 
It measures Euclidean distance between rectangles $r$ and $r'$.
We take the average over all pairs, and normalize by dividing with the $L_1$ distance of the largest diagonal of the two \demersabbrevs.
The result is related to our optimization term for quality when dealing with multiple weights (see Section~\ref{sec:timevar}). 
The second metric, relative directions between layouts (\textbf{SREL}), focuses on changes in relative directions; it is analogous to MREL, but compares two layouts instead.

\paragraph{Datasets.}
We run experiments on real-world datasets. For time-series data, we expect a gradual change and strong correlation between the different values. For weight-vectors data, we expect more erratic changes and less correlation.
We use two maps with rather different geographic structures: the first (\textbf{World}) is a map of world countries, having mixed region (country) sizes in a rather unstructured manner; the second (\textbf{US}) is a map of the 48 contiguous US states, having relatively high structure in  sizes of its states, with large states in the middle and along the west coast and many smaller states along the east coast.
We collected five time series for the World and four for the US map of which the details are given in Appendix~\ref{app:data}.
We transformed these into a weight-vectors dataset by taking the values of 2016 for each of these time series, resulting in five weight vectors for the World map, and four for the US map.

The various datasets have different scales, and need be projected into a reasonable square size to compute a \demersabbrev.
We compute the diagonal $\Delta$ of the bounding box of the map.
For a time-series dataset, we find the region $r$ with maximal $w_i(r)$ for any $i$ and scale values such that $w_i(r) = \Delta/4$.
For a weight-vectors dataset, we do the same, but scale the values for each \demersabbrev separately.

\paragraph{Running times.}
We ran the experiments using IBM ILOG CPLEX 12.8 to solve the (I)LP. 
We observe the following running times on a normal laptop: *-*-IT and FRC-O-* finished within seconds (USA) or a minute (World); *-*-\{SU,CO\} took around a minute (USA) or below 5 minutes (World); 
FRC-T-* was completed in minutes (USA) or hours (World). 
CNT-*-* is an an \emph{integer} linear program rather than a regular linear program (or force-directed method); its computational complexity is significantly higher, and intractable in many cases.
Only CNT-*-IT variants were successfully solved, and only on the US map; for all other cases it ran out of memory (48 GB allocated).

\section{Experimental results}
\label{sec:results}
We discuss results and four questions:
1) How much does the strong versus weak setting affect quality?
2) How much does stability implementation matter?
3) Which optimization criteria perform best?
4) What is the effect of separation constraints in our LP, compared to a force-directed method for \demersabbrevs? 
Fig.~\ref{fig:USdiffalgo} shows the result of two algorithms for the US. Appendix~\ref{app:morefigures} and the supplementary video show more \demersabbrevs for different settings.
%
\begin{figure}[t]
\centering
\includegraphics[width=0.3\linewidth]{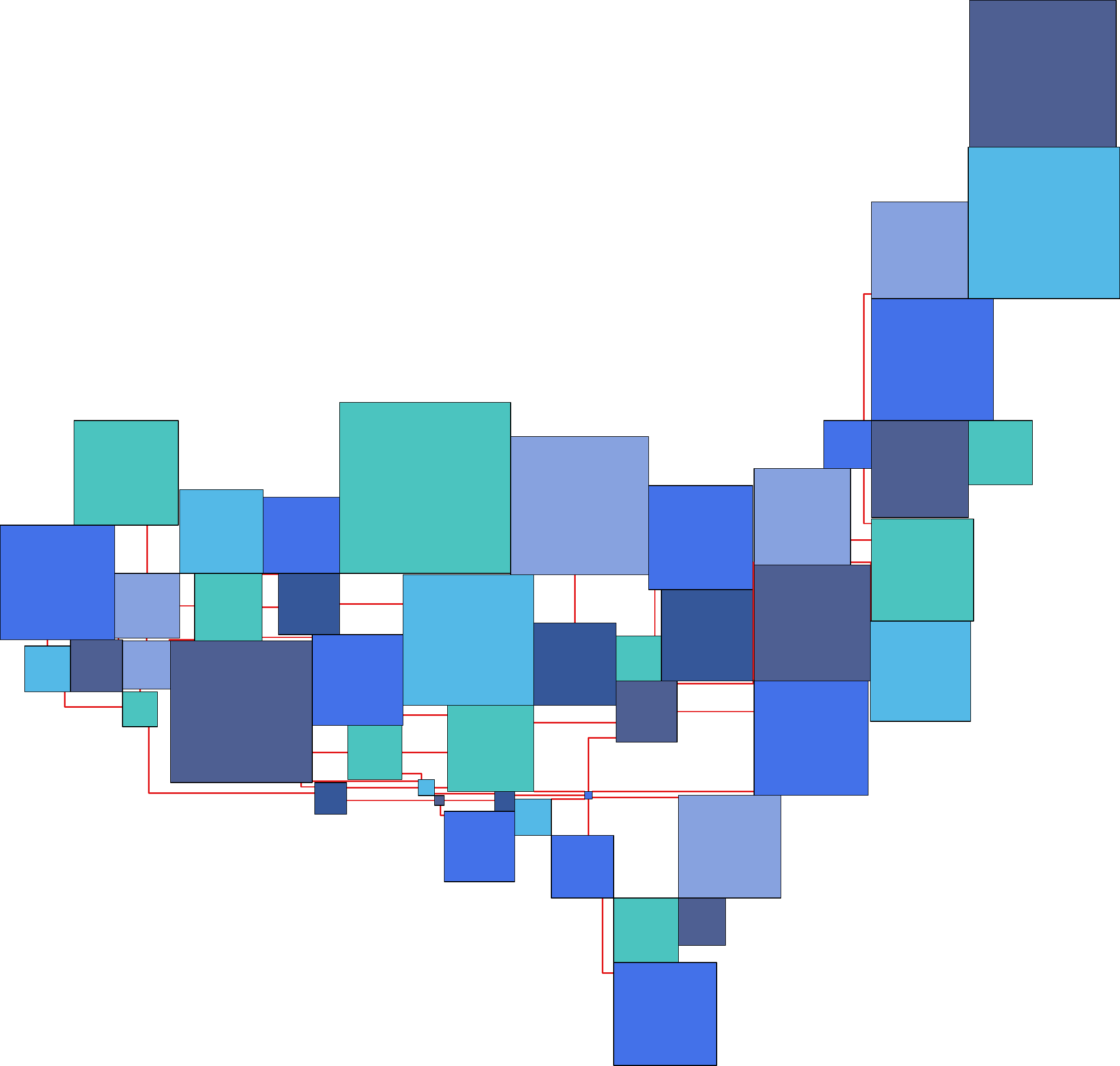}
\quad
\includegraphics[width=0.3\linewidth]{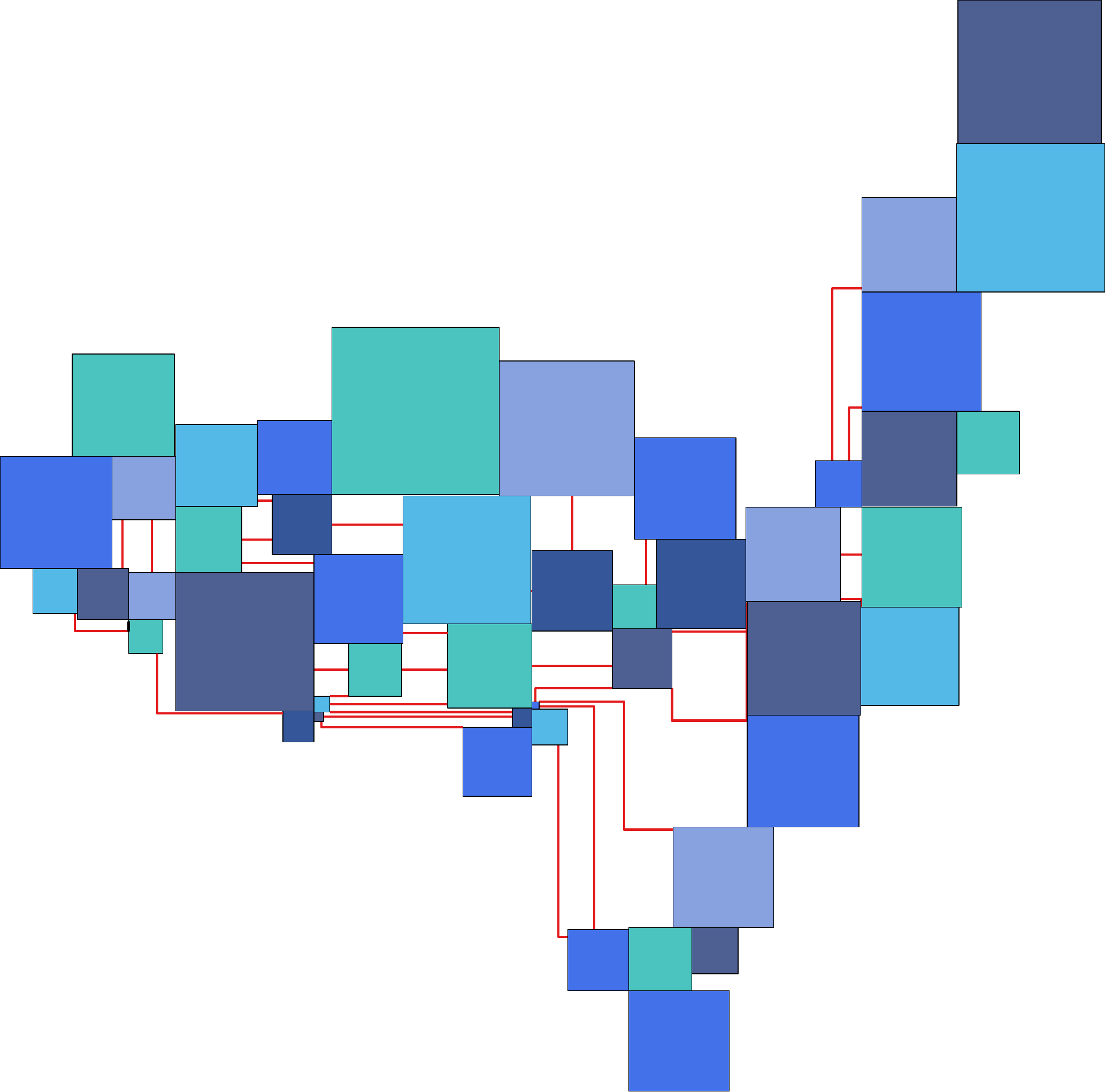}
\caption{ 
US election turnout data \demersabbrev in 2016, by TOP-W-CO and CNT-W-IT.}
\label{fig:USdiffalgo}
\end{figure}

\paragraph{Strong versus weak setting.}
Fig.~\ref{fig:sVSw} shows the average metric values for the iterative variants, over all datasets and linear programs.
We find that the strong case (additional separation constraints) reduces the error in relative direction for both cartogram quality and stability: the average score for MREL, including CNT variants where possible, reduces from 0.21 to 0.16; similarly, stability (SREL) decreases from 0.059 to 0.045 due to decreased movement freedom of the squares. This is at the expense of topological error (MADJ increases from 0.58 to 0.61) and origin displacement (MDIS increases from 0.16 to 0.17).
The effect is present independent of optimization criterion and stability implementation though its strength varies. 
Effects remain noticeable but of varying strength when we control for type of dataset, except MDIS slightly decreases for US datasets (0.116 to 0.107) in the strong setting.
We also see a clear difference between optimization terms (CNT, TOP, ORG), discussed later.

\begin{figure}[t]
    \centering
    \includegraphics[page=2]{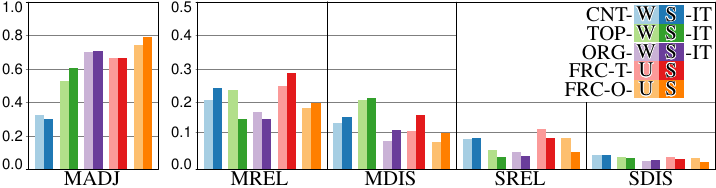}
    \caption{Bar chart of average metric scores, for IT settings of all linear programs and the FRC directed variants. We see similar effects when switching from the weak version to the strong version of the IT setting for all three optimization settings. We also see a strong effect when choosing different optimization settings for the IT setting.  FRC is generally outperformed by the \{TOP,ORG\}-W-IT variants. 
    }
    \label{fig:sVSw}
\end{figure}

\paragraph{Stability implementation.}
In time-series datasets there is little difference in stability over the three settings: time series data change gradually over time so choosing which pairs to optimize does not have much influence. In weight-vectors datasets, even with only few weights per region (five for the World, four for the US), an effect becomes noticeable in the IT setting. CO and SU behave nearly identical, but this might be an artifact of only having a few weights per region.
Compared to CO (and SU) setting, the iterative version scores better on MDIS (0.31 versus 0.26) but worse on the stability metric SDIS (0.084 versus 0.10). 
For weight-vectors datasets it is thus better to use the SU variant as this achieves better stability and is only slightly more expensive to compute compared to IT variants. The added complexity of CO does not seem to pay off. 

\paragraph{Optimization criteria.}
We use three metrics for cartogram quality: MADJ and MDIS are optimized explicitly with the CNT and ORG objectives respectively, the third metric MREL corresponds to a secondary objective term.
To compare the TOP/CNT/ORG objective terms, we consider the IT variant (see Fig.~\ref{fig:sVSw}), as other stability implementations could not solve the CNT objective; still, we found similar patterns for the SU and CO cases.

For MADJ, CNT finds the optimal value (0.31) under the given constraints. 
TOP (0.57) does clearly better than ORG (0.70), somewhat in contrast to observations of \cite{eppstein,smwg}: for grid maps, the MDIS metric that ORG optimizes is a good proxy for maintaining topology; our results suggest this is not so for \demersabbrevs.

For MDIS and MREL metrics ORG performs best; for MDIS, CNT performs slightly better compared to TOP and vice versa for MREL. 
Thus, in terms of spatial quality, ORG seems a good objective, except for topological error -- which is typically of primary concern for cartograms.

For stability metrics SDIS and SREL, ORG outperforms TOP which outperforms CNT.
We explain it by inherent stability of the map which is the same for all \demersabbrevs. CNT does poorly; it is fairly unconstrained for lost adjacencies whereas TOP aims to keep such pairs close.

ORG scores best on all metrics except MADJ; its MADJ score is high, losing 70\% of adjacencies on average. 
In contrast, CNT optimizes the number of adjacencies, but is clearly worse on other metrics and is computationally expensive. 
There is thus a trade-off present between topological error and other quality aspects. 
TOP makes this trade-off, scoring reasonably on most metrics.

\paragraph{Comparison to FRC.}
Our linear programs enforce separation constraints which help maintain spatial relations and the spatial mental model; they are required for the linear program but not in general.
To study their effect, we compare to FRC which does not enforce separation constraints; results are shown in Fig.~\ref{fig:sVSw}. 
Comparing FRC-T and FRC-O variants, we see the same behavior as in 
the TOP versus ORG linear programs:
FRC-O performs worse than FRC-T on ADJ, and better on the other metrics.
Layout initialization trades off stability versus cartogram quality:
FRC-*-S variants have better stability scores and worse quality scores compared to FRC-*-U.

As it has the fewest constraints, we compare ORG-W-IT to FRC methods:
FRC-O-* are slightly worse or equal to ORG-W-IT on all metrics; 
FRC-T-* are worse than ORG-W-IT on all metrics except ADJ
where it is a lightly better, but the number of adjacencies lost is still clearly higher compared to TOP-W-IT.

To conclude, in general we outperform FRC for the various metrics by an appropriate setting in our linear program.
No single setting outperforms all FRC variants.
The large difference with TOP-variants in terms of MADJ suggests TOP variants are a good choice for high-quality stable \demersabbrevs.

\section{Discussion and future work}

We described a linear program to compute stable Demers cartograms, based on separation constraints and minimizing distance between adjacent regions. It allows overlap-free transitions between weight functions and connecting lost adjacencies with short, low-complexity leaders.
Experiments show it offers a good trade-off between topological error and other criteria.
It outperforms basic force-directed layouts, though there is not a unique variant that does so, suggesting an interplay between separation constraints, optimization, and quality metrics.

In future work we may consider stability in other cartogram styles, and perform human-centered comparisons in addition to computational ones, with methods implemented in interactive systems; such systems can, e.g., emphasize adjacent regions by drawing leaders (at all or more clearly) or link regions back to the geographic map.
We focused on Demers cartograms, but there are many different styles of cartograms. Future work may also investigate stable variants of such other cartogram styles and quantitatively or qualitatively compare them.

\bibliographystyle{splncs04}

\bibliography{references}


\clearpage

\appendix

\section{Omitted proofs}
\label{app:leaderlength}
\label{app:leaderbends}

\leaderlength*
\begin{proof}
If the regions map to touching squares in $B$, the claim clearly holds, so assume that the squares are disjoint.
Our assumption of minimality tells us that there is a cartogram $A$ with $L_1^A(r_1,r_2) = 0$.
Without loss of generality, assume that $(r_1,r_2) \in H$ (thus $r_1$ is left of $r_2$ in both $A$ and $B$). The other cases are analogous up to rotation and symmetry. 
This assumption implies that the common boundary between $r_1$ and $r_2$ in $A$ is a vertical segment.
We have to consider three cases for the positions of $r_2$ with respect to $r_1$ in $B$.

\emph{Case 1:}
A horizontal line intersects both $r_1$ and $r_2$.
We draw a horizontal leader $\ell$ directly from $r_1$ to $r_2$ in $B$ along this line.
To derive a contradiction, assume that some other region $r$ in $B$ intersects $\ell$. A horizontal separating line can then exist neither with $r_1$ nor $r_2$, and thus $(r_1,r),(r,r_2) \in H$. However, $r$ is also between $r_1$ and $r_2$ in $A$: a contradiction to $r_1$ and $r_2$ touching in $A$.

\begin{figure}[b]
    \centering
    \includegraphics{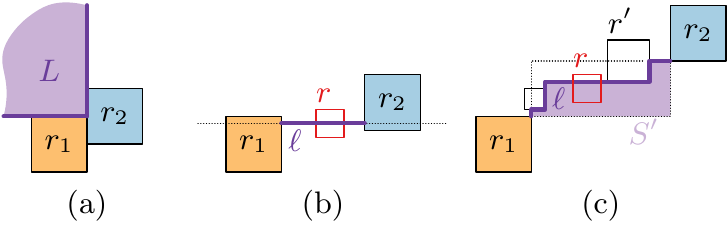}
    \caption{(a) Two adjacent regions in cartogram $A$. For case 2, the region $L$ defining $S'$ lies above $r_1$ and to the left of $r_2$. (b) Case 1 in cartogram $B$: there is a horizontal line intersecting both regions (dotted). Intersecting region $r$ cannot exist. (c) Case 2 in cartogram $B$: $r_2$ is above $r_1$. A region $r$ intersecting $\ell$, the upper boundary of $S'$, cannot exist.}
    \label{fig:my_label}
\end{figure}

\emph{Case 2:}
$r_2$ is fully above $r_1$ in $B$.
Let $S = [x_1,x_2] \times [y_1,y_2]$ denote the rectangle from the top-right corner of $r_1$ to the bottom-left corner of $r_2$.
We must prove that $S$ contains a monotone leader from $(x_1,y_1)$ to $(x_2,y_2)$, such that it does not intersect any square in $B$.
Let $L$ contain all regions $r$ for which $(r_1,r) \in V$ and $(r,r_2) \in H$.  
Consider a domain $S' \subseteq S$ that is defined by removing all points $(x,y)$ from $S$, for which there is a region $r \in L$ with bottom-right corner $(p,q)$ in $B$ such that $x < p$ and $y < q$.
Note that $S'$ cannot be empty since the separation constraints for the regions in $L$ ensure that the right sides of their squares are strictly to the left of $x_2$ and strictly above $y_1$.
The upper boundary of $S'$ is an orthogonal leader $\ell$ of length $L_1^B(r_1,r_2)$.
Note that each concave bend of $\ell$ (with respect to $S'$) corresponds to some region $r \in L$.
As the upper boundary is trivially orthogonal and of length $L_1^B(r_1,r_2)$, we need to argue only that no other square can intersect $\ell$.

For a contradiction, assume that there is some region $r$ with a square in $B$ that intersects $\ell$.
Note that $r$ cannot be in $L$, as it would otherwise contribute to $S'$.
It must either intersect a horizontal or a vertical segment of $\ell$. Assume that it intersects a horizontal segment. Let $r'$ be the region in $L$ corresponding to the concave bend at the (right) end of the segment; if there is no bend point, we use $r' = r_2$ instead. Since $r$ and $r'$ only have a vertical separator, we know that $(r,r') \in H$. Since $(r',r_2) \in H$ by definition of $L$ (or $r' = r_2$), we know that $r$ is to the left of $r_2$ in $A$.
If $(r_1,r) \in V$, then we would have a contradiction as $r$ would be in $L$. Thus, assume that this is not the case. This implies that $(r_1,r)\in H$, as its position in $S'$ disallows $(r,r_1)\in V$. Thus, $r$ is to the right of $r_1$ in $A$. However, since $r_1$ and $r_2$ are adjacent in $A$, there cannot be such a square to the left of $r_2$ and to the right of $r_1$ and we arrive at a contradiction. An analogous argument shows that no square can intersect a vertical segment. Thus, $r$ cannot exist, and $\ell$ is a leader that does not cross any square.

\emph{Case 3:} $r_2$ is fully below $r_1$. We apply an argument similar to Case 2, defining $S'$ using the regions that are above $r_2$ and to the right of $r_1$ instead, and removing everything above and to the right of the bottom-left point of such regions.

Thus, in all cases cartogram $B$ admits a leader of length $L_1^B(r_1,r_2)$ between two regions that are in $T$, if a cartogram $A$ exists that realizes this adjacency.\hfill\qed
\end{proof}

Based on this proof, the construction of such a minimal leader is straightforward. After sorting the regions on $y_r-w(r)/2$ (the bottom side), we iterate through these regions from bottom to top, and incrementally maintain the rightmost coordinate of a region in $L$, which then traces the upper boundary of $S'$. This is for the illustrated case, the other three cases are analogous. Assuming a planar (and thus linear-size) $G$, it takes $O(n^2)$ time to compute \emph{all} leaders, since the sorting is done only a constant number of times.

\leaderbends*
\begin{proof}
We follow the same setup and case distinction as for the proof of Lemma~\ref{lem:leaderlength}.
In Case 1, the construction readily leads to a 0-bend leader. Cases 2 and 3 are analogous; we therefore consider only case 2 here, where $r_2$ is fully above $r_1$.

Again, consider the region $S = [x_1,x_2] \times [y_1,y_2]$ between the closest corners of the two regions, and the set $L$. Now, we construct the following leader: go up from $(x_1,y_1)$ until you reach either $y_2$ or the $y$-coordinate of the bottom of a region in $L$ (the first convex bend of $S'$ in the proof of Lemma~\ref{lem:leaderlength}). From here, we go to the right, until we reach $x_2$ and go up (if we did not reach $y_2$ yet) to finalize the leader. This leader $\ell$ has at most 2 bends by construction. It remains to argue that no region in $B$ intersects $\ell$. We assume for ease of exposition that the leader consists properly of three segments.

\begin{figure}[t]
    \centering
    \includegraphics{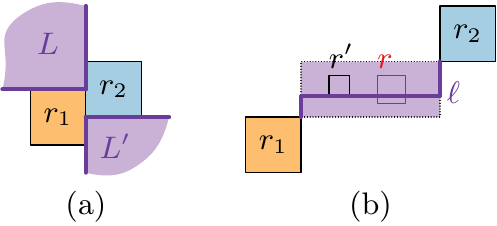}
    \caption{(a) Two adjacent regions in cartogram $A$. (b) A region $r$ that intersects $\ell$ to the right of the defining region $r'$. $r'$ must lie in $L$ and $r$ in $L'$ in $A$, leading to a contradiction.}
    \label{fig:leaderBends}
\end{figure}

Another region can neither intersect the first vertical segment, nor the horizontal segment to the left of its defining square, by the proof of Lemma~\ref{lem:leaderlength}.

Assume for a contradiction that a region $r$ intersects the horizontal segment (or the second vertical segment), position to the right of its defining region $r'$. 
This placement implies that $(r_1,r) \in H$, thus $r$ is to the right of $r_1$ in $A$. Moreover, $r$ is horizontally or vertically separated from $r_2$ (or both). In the former case, it is thus to the left of $r_2$: this is impossible since $r_1$ and $r_2$ are adjacent in $A$. In the latter case, $r$ is in the region $L'$ in Fig.~\ref{fig:leaderBends}. We observe that $r'$ must be fully above $r$ in $A$ and thus $(r,r') \in V$ by definition of the strong separation constraints. 
Thus, the placement in $B$ does not satisfy the separation constraints and we have a contradiction.\hfill\qed
\end{proof}

\section{Force-directed implementation}
\label{app:force}

For each region $r$, we define a force $f(r)$ determining 
the
direction it should move, as a sum of two forces $f(r) := f_d(r) + f_q(r)$.
The \emph{disjointness} force $f_d(r) := D \cdot \sum_{r'} f_d(r,r')$ tries to make $r$ disjoint from other regions; every other region $r'$ yields a force on $r$ if their squares overlap, so $f_d(r,r')$ has direction $r - r'$ and magnitude $((m - L_\infty(r,r')) / m)^2$, where $L_\infty(r,r')$ is the Chebyshev distance between 
square centers
and $m:=w_{r,r'} + gap_{r,r'}$
(cf.\ Section~\ref{sec:single}).
To ensure (nearly) disjoint squares when forces are in equilibrium, we multiply $D$ by a constant.\footnote{50\,000 in our implementation.}
We add another force to capture \emph{cartogram quality}, with two variants; 
we either use forces attracting a region to its origin (centroid in the geographic map), or to  regions to which it should be adjacent according to $T$.
For the origin-attracting version, define $f_q(r)$ with direction $o - r$ and magnitude $\| o - r\| / \Delta$, with $o$ the origin of region $r$ and $\Delta$ the diagonal of the bounding box of all origins.

For the topology-attracting version, we define $f_q(r)$ as the average over all $r'$ with $\{r,r'\} \in T$ of the vector with direction $r' - r$ and magnitude $(L_\infty(r,r') - m)/m$;
if squares of the regions overlap, we set this vector to zero.
We iteratively apply the forces, scaled to ensure no force is larger than the minimum weight (to avoid a square jumping over another). 
We iterate until forces are small enough \footnote{In our implementation this is  $10^{-5}$ times the minimum weight} so we are close or at an equilibrium; note it may leave squares with slight overlap, we ignore such imperfections as allowing some overlap  extends the solution space and should not decrease quality by the metrics below.
This computes a \demersabbrev~for one set of weight values, and must be rerun for each value set. 
We initialize locations for the first set of weight values using  centroids in the map (origins); for subsequent runs, we may choose the same initialization, or centers of the previous layout to improve stability.

\section{Data sources}
\label{app:data}

\stepcounter{footnote}
\footnotetext{\label{f1}\url{https://en.wikipedia.org/wiki/Economy_of_North_Korea\#cite_note-2}}
\stepcounter{footnote}
\footnotetext{\label{f2}\url{https://tradingeconomics.com/kuwait/rural-population-wb-data.html}}
\stepcounter{footnote}
\footnotetext{\label{f3}\url{http://www.efgs.info/wp-content/uploads/conferences/efgs/2016/S8-1_presentationV1_IdrizShala_EFGS2016.pdf}}
\stepcounter{footnote}
\footnotetext{\label{f4}\url{https://tradingeconomics.com/eritrea/rural-population-wb-data.html}}
\stepcounter{footnote}
\footnotetext{\label{f5}\url{https://tradingeconomics.com/kuwait/forest-area-percent-of-land-area-wb-data.html}}
\stepcounter{footnote}
\footnotetext{\label{f6}\url{https://tradingeconomics.com/kosovo/forest-area-percent-of-land-area-wb-data.html}}

The datasets used for the world map are taken from the World Bank~\cite{worldbankdata}. Not all time series are quite complete for these datasets, to avoid biasing results or exaggerating instability, we augmented the data using other sources as is given in Table~\ref{tab:augmentData}. All websites used for this\textsuperscript{\ref{f1}--\ref{f6}} were accessed in February 2019.

For the USA datasets we used four different datasets.
The Drug Poisoning Mortality dataset from the Center for Disease Control and Prevention~\cite{drugdata};
the General Election Turnout dataset is from the United States Elections Project~\cite{electionsdata}; the GDP dataset is from the US Bureau of Economic Analysis~\cite{GDPdata}; and the US population dataset is from the US Census Bureau~\cite{populationdata}.
The complete details for the datasets collected are shown in Table~\ref{tab:dataSetTable}.

\begin{table}[htbp]
    \centering
    \caption{Time-series datasets used in our experiments.}
    \label{tab:dataSetTable}
    \begin{tabular}{lll|lll}
        \textbf{Map} & \textbf{Time series} & \textbf{Years} & \textbf{Map} & \textbf{Time series} & \textbf{Years}\\
        \hline
        World & Forest Area & 2006--2016 & US & Drug Poisoning Mortality & 2007--2016\\ 
        World & GDP & 2006--2016  & US & GDP & 2007--2016\\
        World & GDP per Capita & 2006--2016  & US & General Election Turnout & 1998--2016  \\
        World & Population Growth & 2006--2016 &    &                          & (even years) \\
        World & Rural Population & 2006--2016 & US & Population & 2011--2020 \\
    \end{tabular}
\end{table}

\begin{table}[htbp]
\label{tab:augmentData}
\centering
\caption{Sources for handling missing data values in the time series.}
    \begin{tabular}{lll}
        \textbf{Time series} & \textbf{Country} & \textbf{Source} \\
        \hline
        GDP & DPR Korea & bloomberg via wikipedia\textsuperscript{\ref{f1}} \\
        Rural population & Kuwait & tradingeconomics.com\textsuperscript{\ref{f2}} \\
        & Kosovo & efgs.info\textsuperscript{\ref{f3}} \\
         & Eritrea & tradingeconomics.com\textsuperscript{\ref{f4}} \\
        Forest area & Kuwait & tradingeconomics.com\textsuperscript{\ref{f5}} \\
          & Kosovo & tradingeconomics.com\textsuperscript{\ref{f6}} \\
         & Eritrea & Extended the 2012 value \\
         population growth & Eritrea & Extended the 2012 value \\
    \end{tabular}
\end{table}

%
%
%
%
%
%

\section{Additional cartograms}
\label{app:morefigures}

The figures on this page and the next illustrate the results of the various algorithms used in our experiments on the US weight-vectors dataset. For each algorithm we show the same three data values: Drug poisoning mortality, population and general election turnout, all in 2016.


\begin{figure}[b]
    \centering
    \includegraphics[width=.3\linewidth]{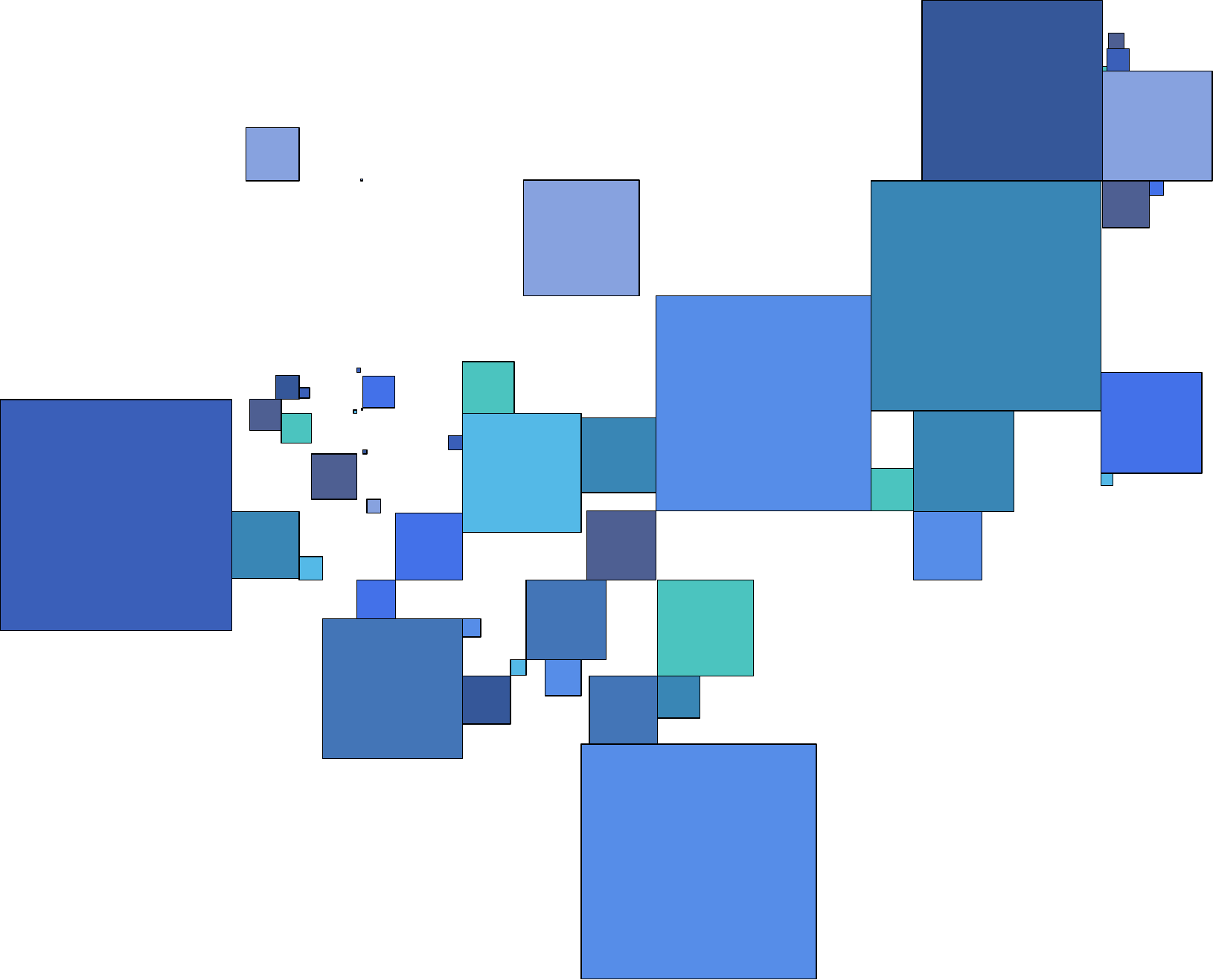}
    \hfill
    \includegraphics[width=.3\linewidth]{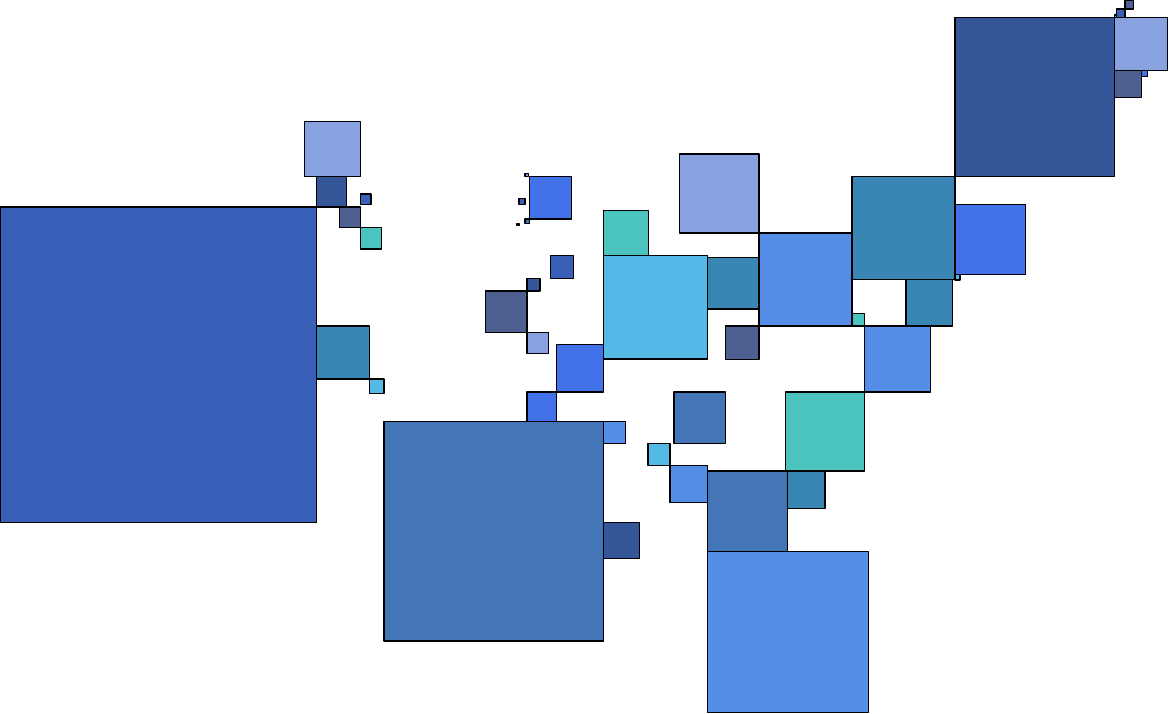}
    \hfill
    \includegraphics[width=.3\linewidth]{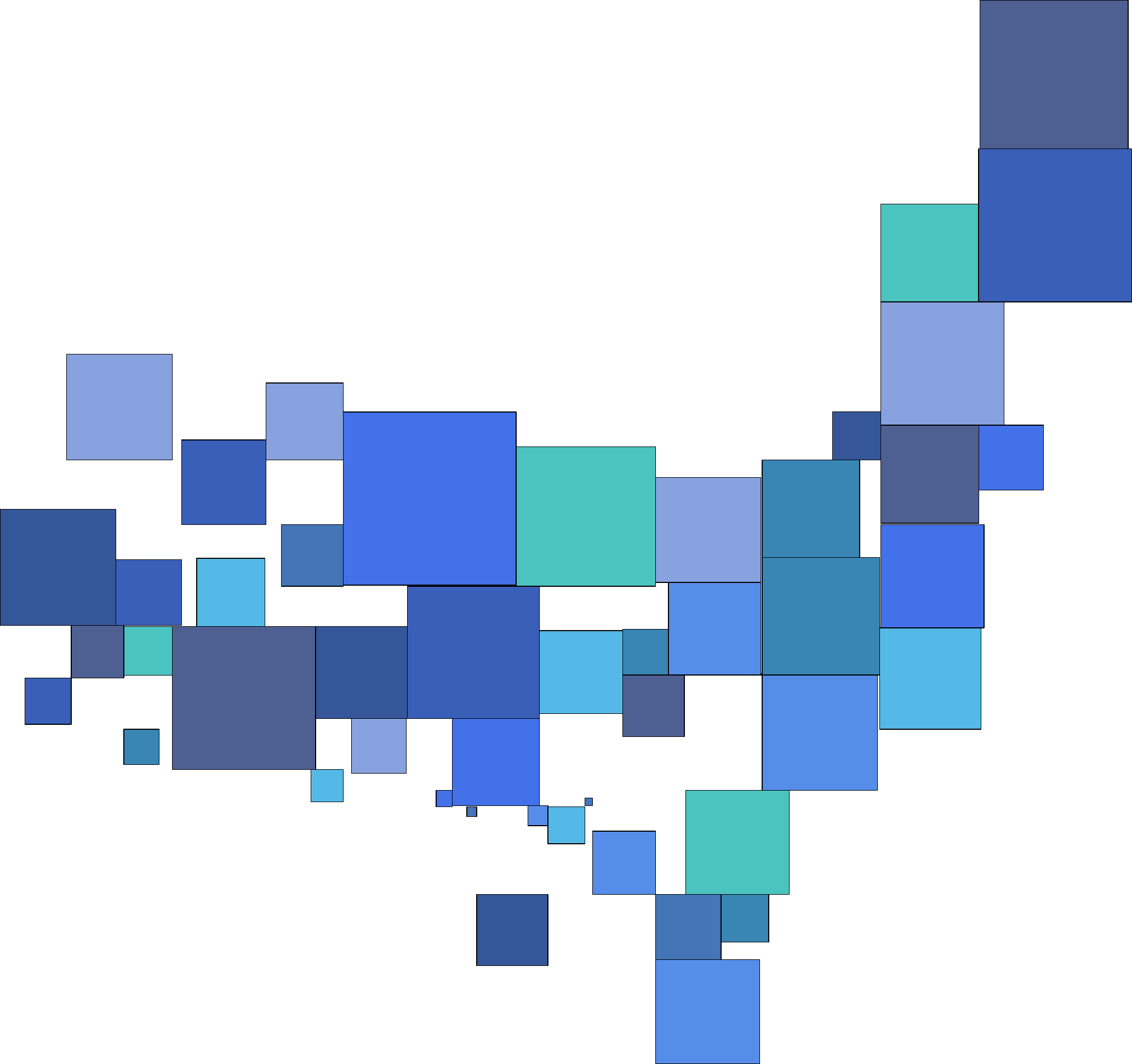}
    \caption{TOP-S-CO: minimizing distance between adjacent regions with strong separation constraints.} 
    \label{fig:USA_strong}
\end{figure}
\begin{figure}[b]
    \centering
    \includegraphics[width=.3\linewidth]{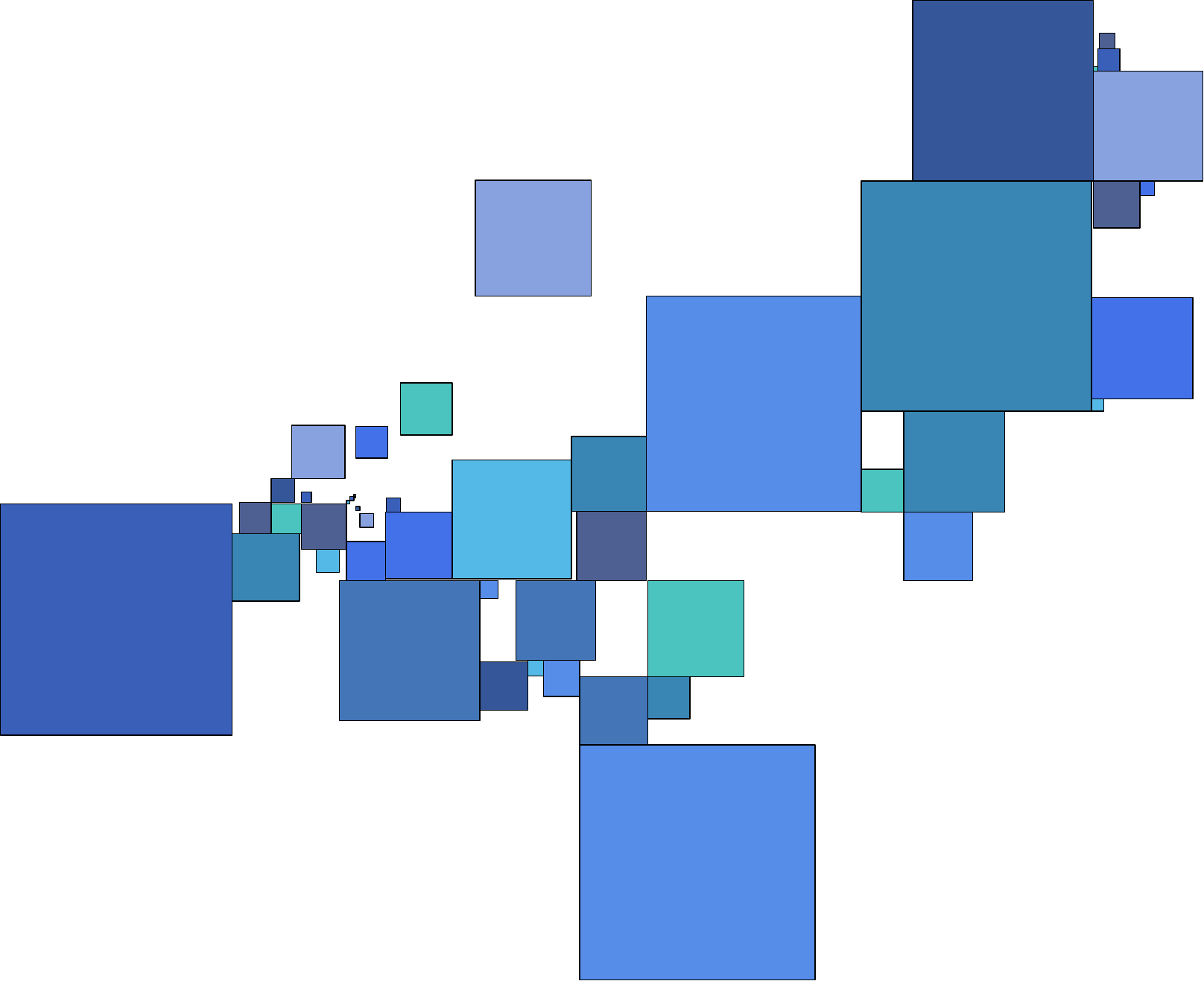}
    \hfill
    \includegraphics[width=.3\linewidth]{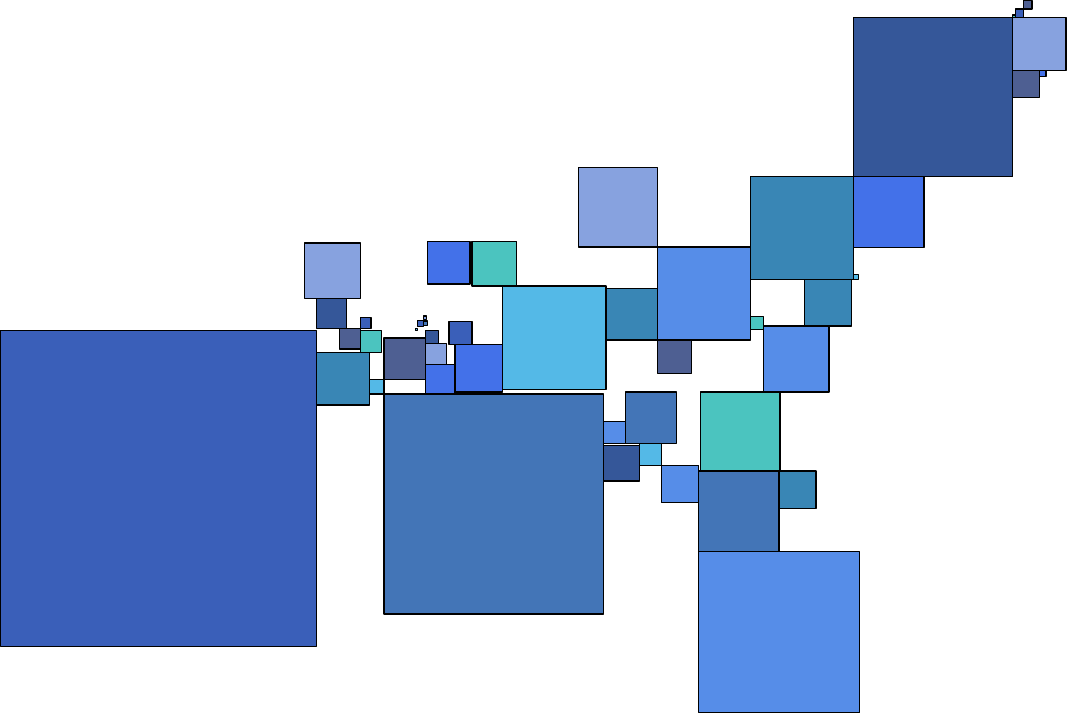}
    \hfill
    \includegraphics[width=.3\linewidth]{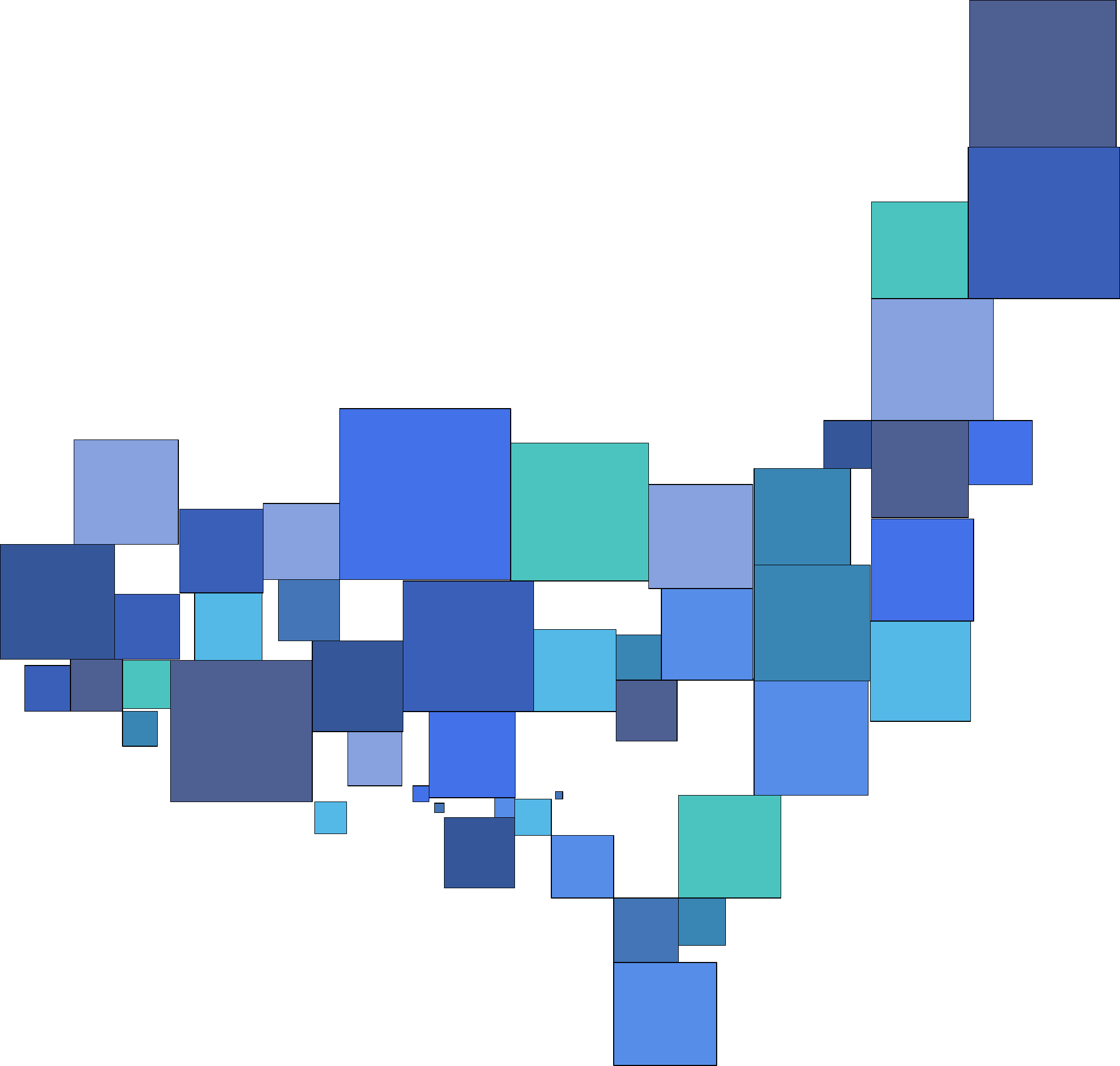}
    \caption{TOP-W-CO: minimizing distance between adjacent regions with weak separation constraints.}
    \label{fig:USA_weak}
\end{figure}
\begin{figure}[b]
    \centering
    \includegraphics[width=.3\linewidth]{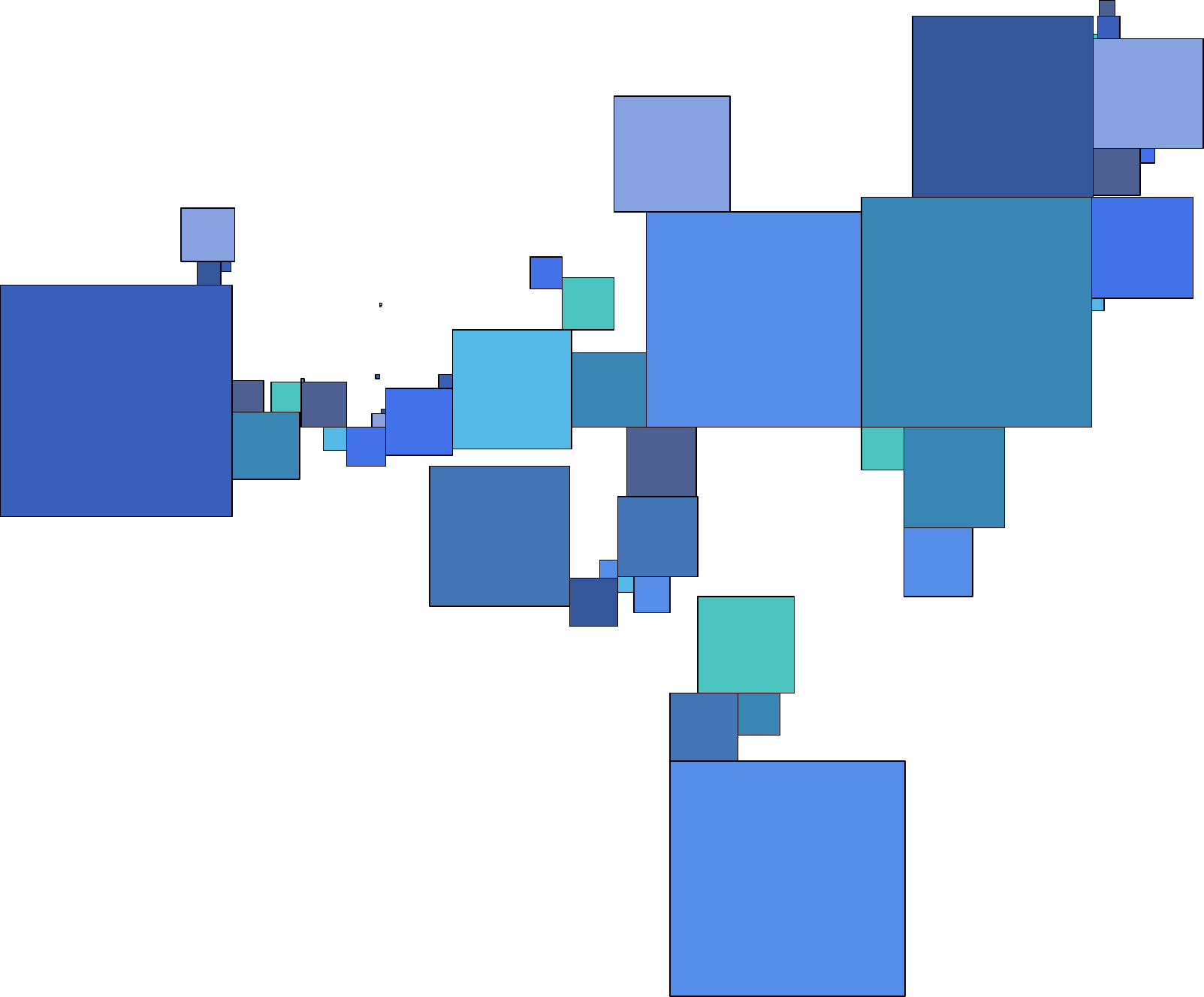}
    \hfill
    \includegraphics[width=.3\linewidth]{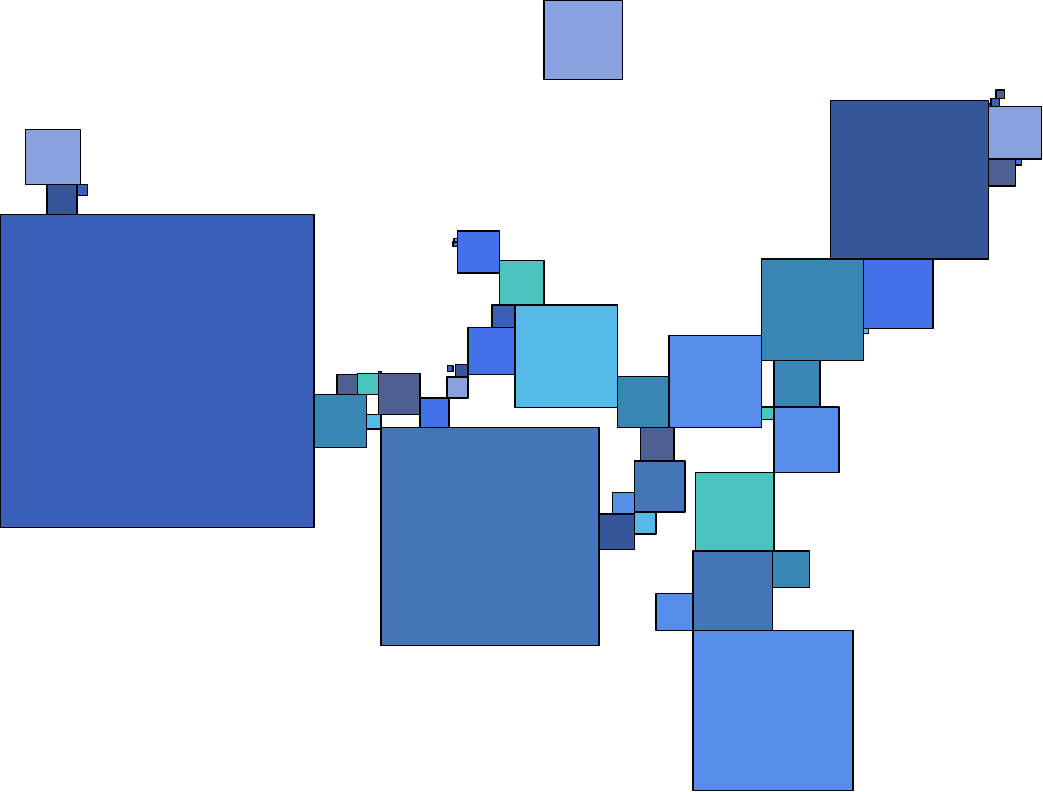}
    \hfill
    \includegraphics[width=.3\linewidth]{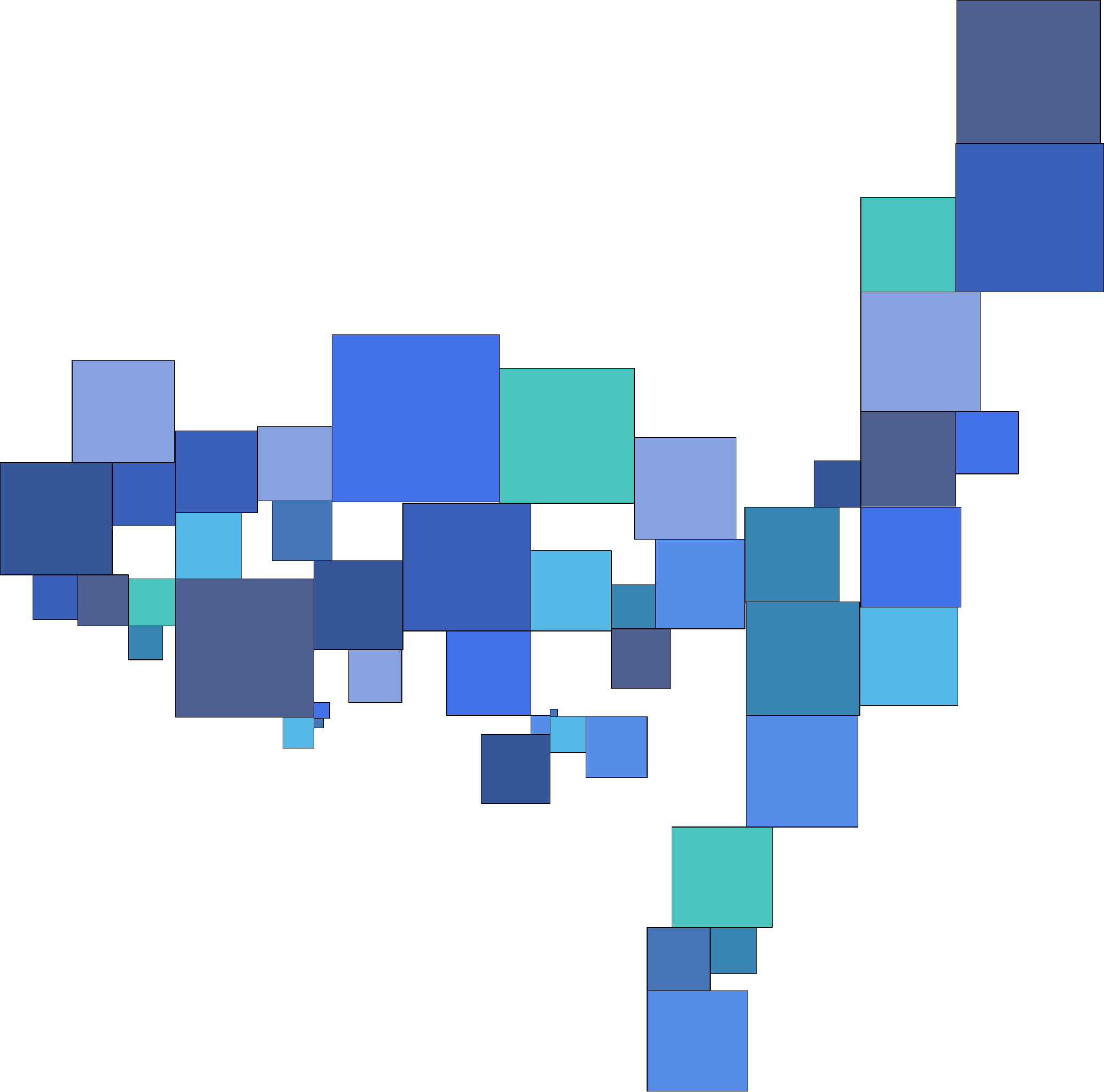}
    \caption{CNT-W-IT: minimizing lost adjacencies with weak separation constraints, using the iterative stability implementation.}
    \label{fig:USA_cnt}
\end{figure}
\begin{figure}[t]
    \centering
    \includegraphics[width=.3\linewidth]{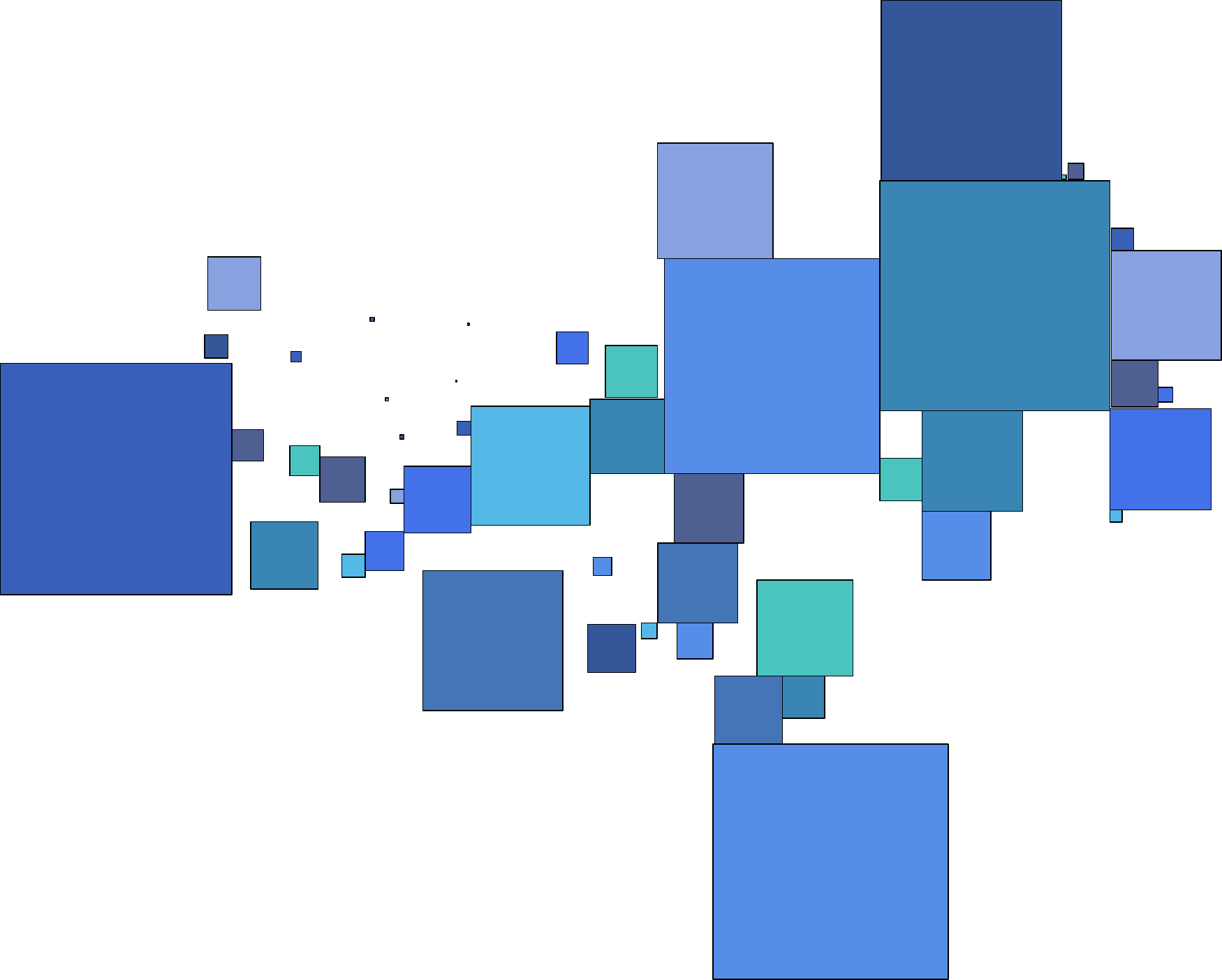}
    \hfill
    \includegraphics[width=.3\linewidth]{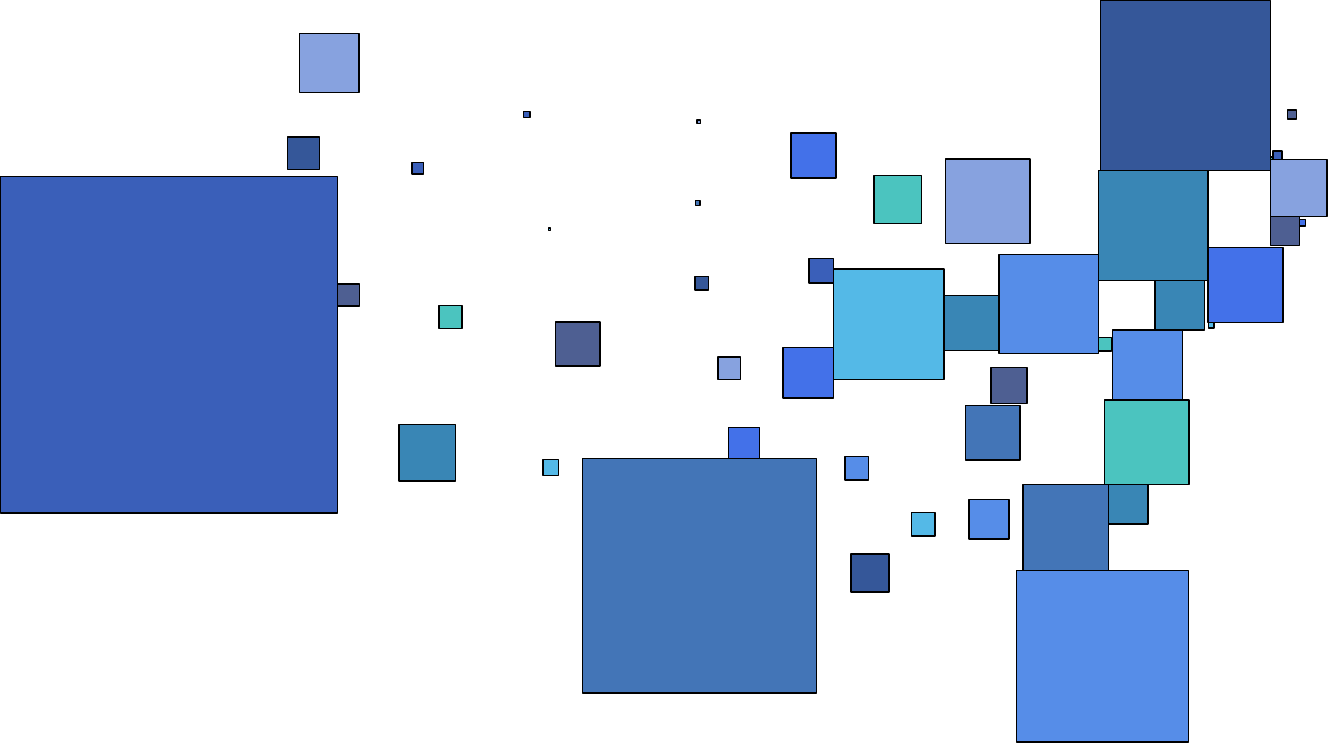}
    \hfill
    \includegraphics[width=.3\linewidth]{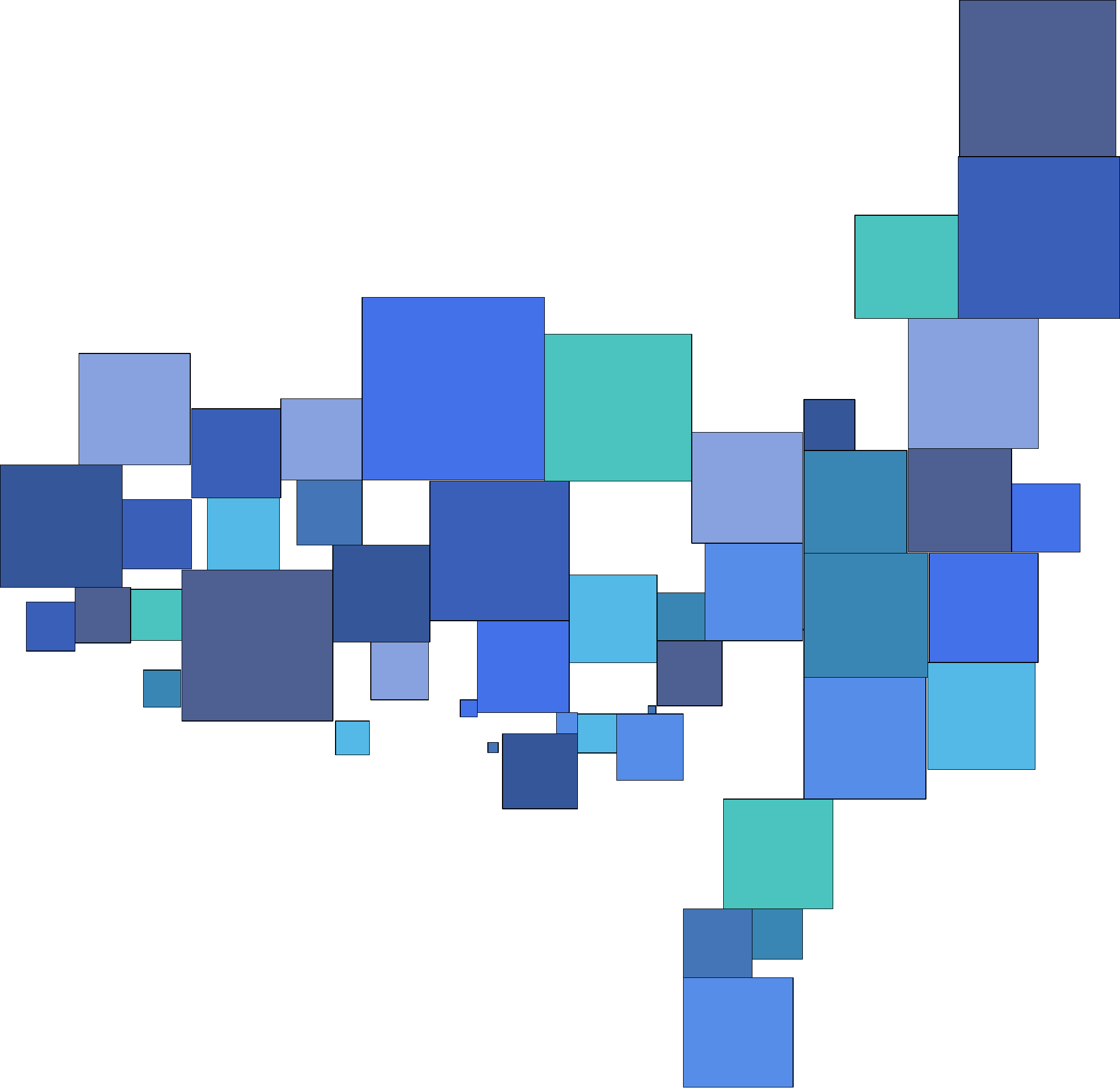}
    \caption{ORG-W-CO: minimizing distance to input position in a weight-vectors dataset and weak setting}
    \label{fig:USA_org}
\end{figure}
\begin{figure}[t]
    \centering
    \includegraphics[width=.3\linewidth]{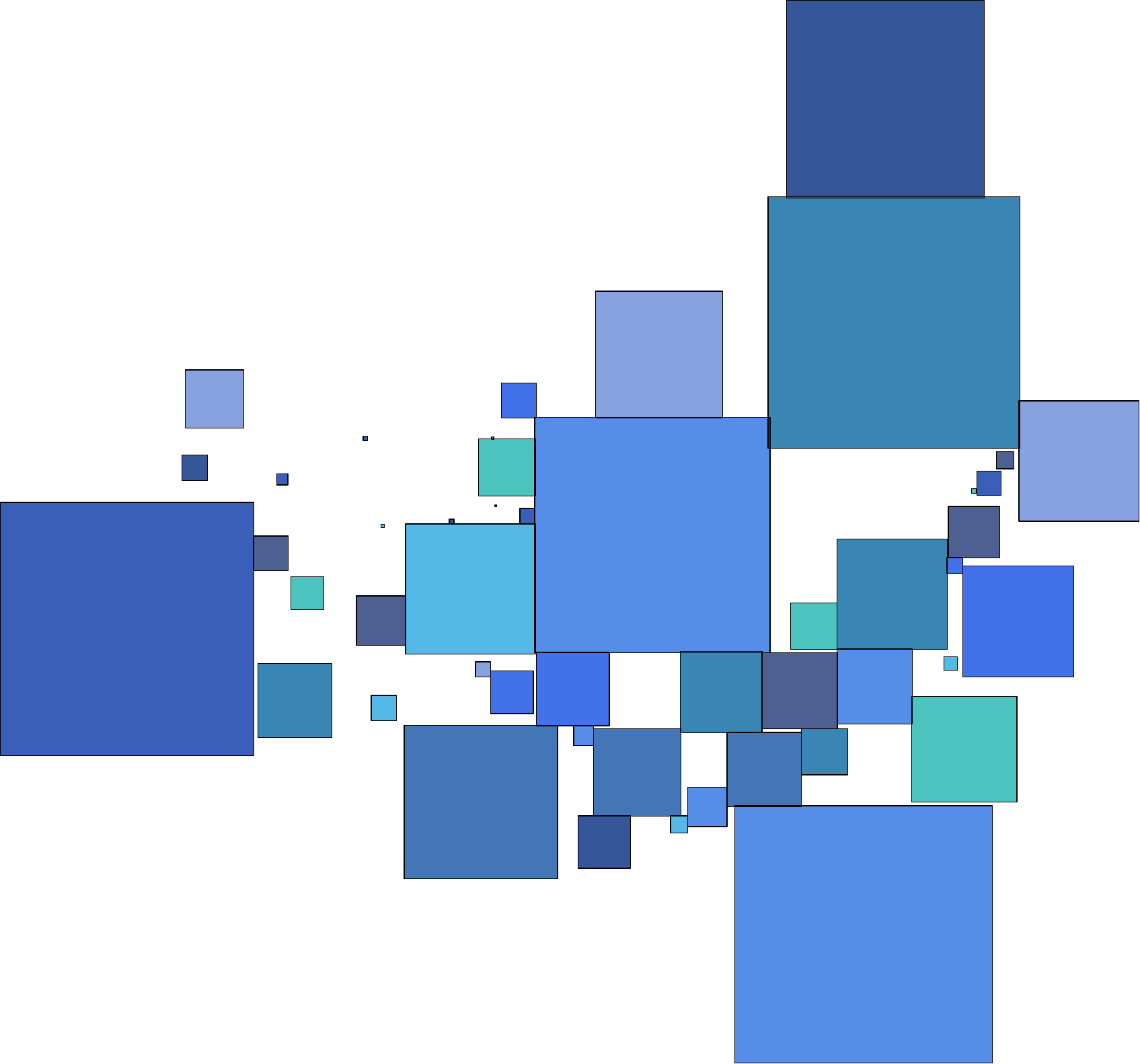}
    \hfill
    \includegraphics[width=.3\linewidth]{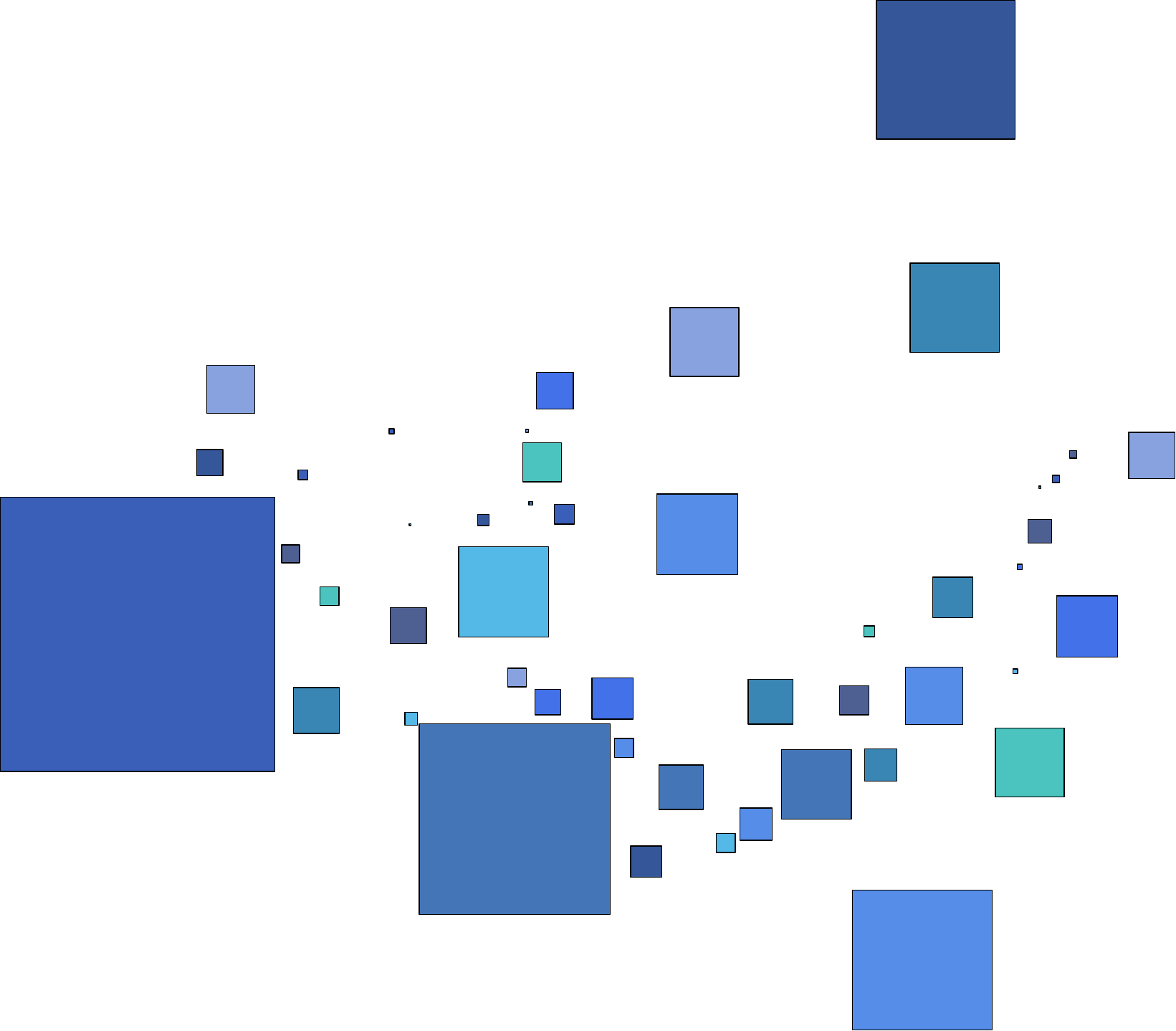}
    \hfill
    \includegraphics[width=.3\linewidth]{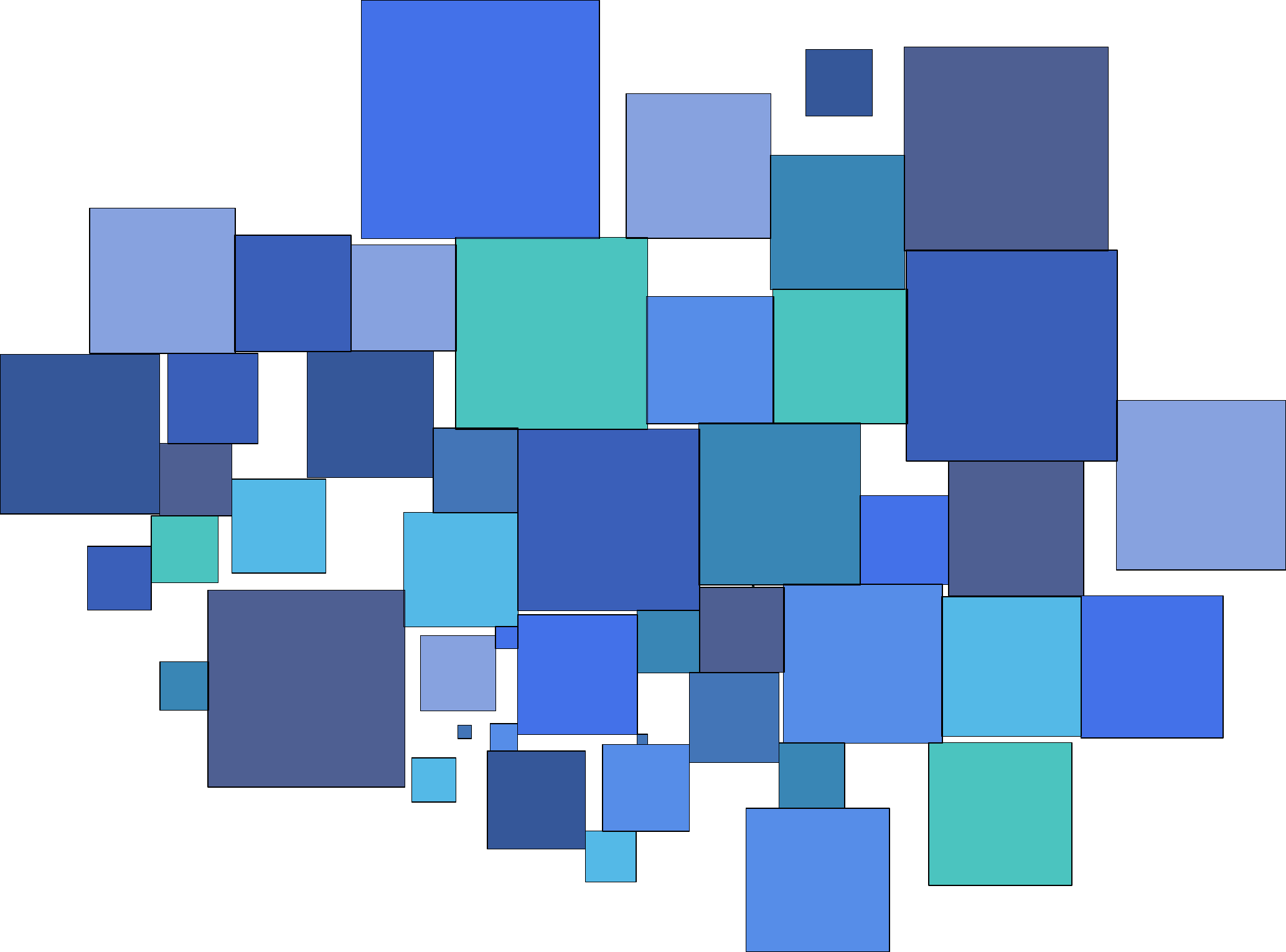}
    \caption{FRC-O-S: minimizing distances to origins through the forced-directed method.}
    \label{fig:USA_force_o}
\end{figure}
\begin{figure}[t]
    \centering
    \includegraphics[width=.3\linewidth]{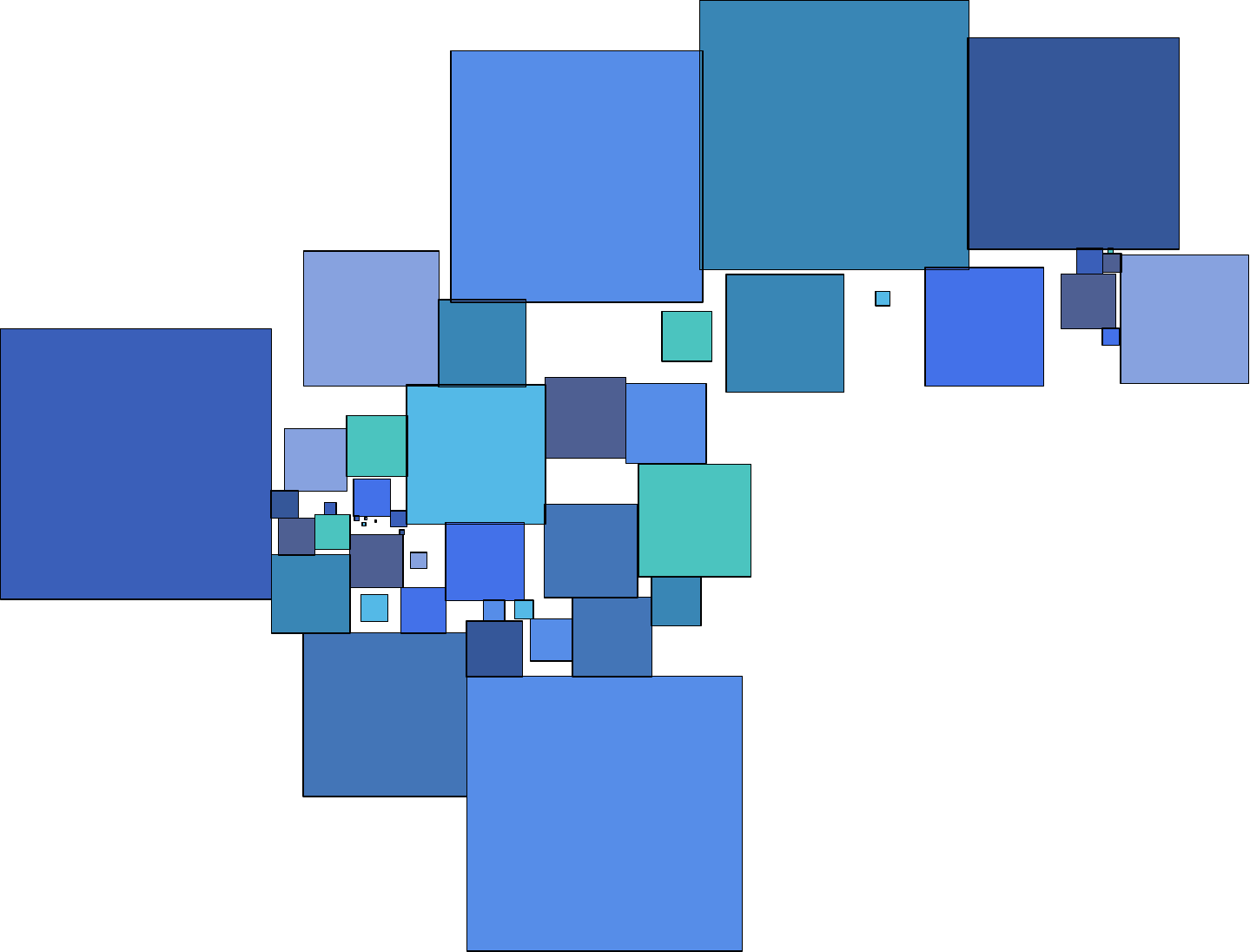}
    \hfill
    \includegraphics[width=.3\linewidth]{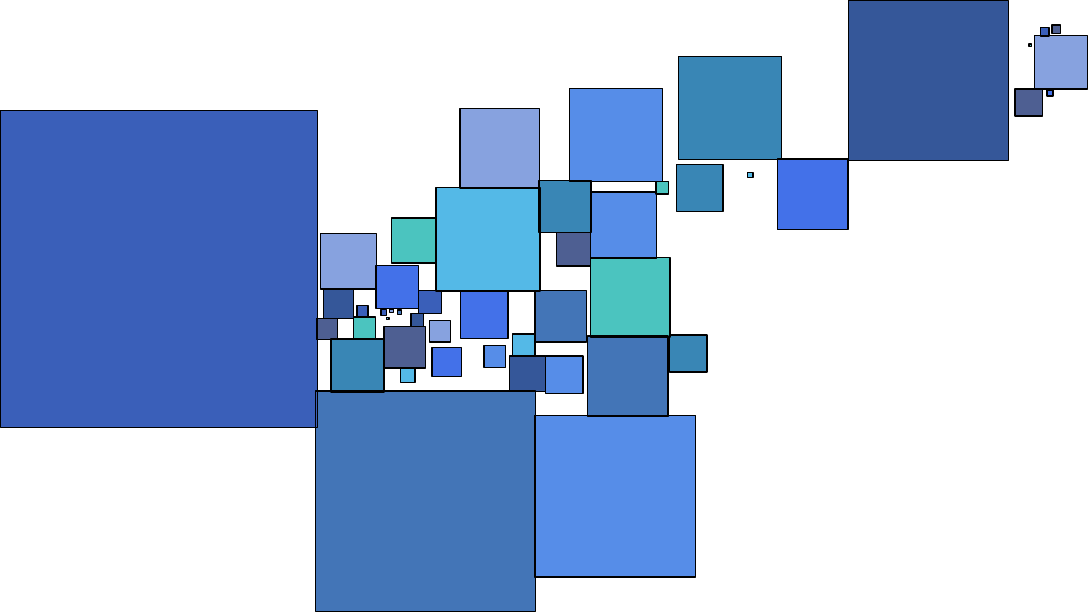}
    \hfill
    \includegraphics[width=.3\linewidth]{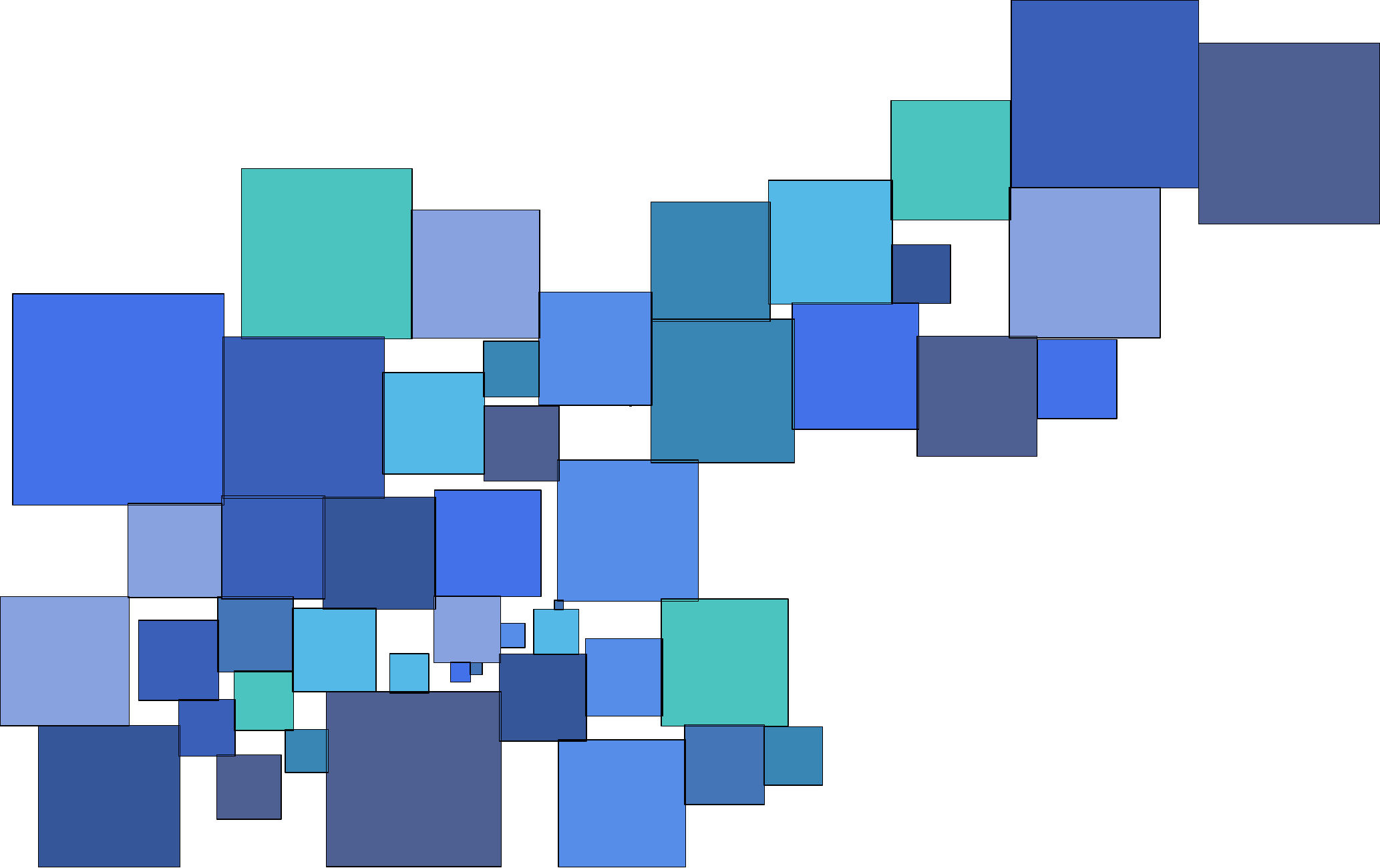}
    \caption{FRC-T-S: minimizing distances between adjacent regions through the force-directed method.}
    \label{fig:USA_force_t}
\end{figure}

\end{document}